\documentclass[fontsize=10pt]{article}
\usepackage[hyphens]{url}
\usepackage[a4paper, total={6in, 8in}]{geometry}
\usepackage{graphicx}
\usepackage{natbib}
\usepackage[font=small,labelfont=bf]{caption}

\usepackage{mathtools}
\usepackage{amsthm}
\usepackage{amsmath}
\usepackage{amsfonts}
\usepackage{amssymb}

\usepackage{hyperref}
\usepackage{cleveref}

\usepackage{algorithm}
\usepackage{algorithmic}

\usepackage{authblk}

\usepackage{preamble}
\def\equalcontrib{%
  \ifnum\value{eqfn}=0%
    \footnote{These authors contributed equally.}%
  \else%
    \footnotemark[\value{eqfn}]%
  \fi%
}%

\title{On the Edge of Core (Non-)Emptiness: An Automated Reasoning Approach to Approval-Based Multi-Winner Voting}
\author[1,2]{Ratip Emin Berker\footnote{These authors contributed equally.}}
\author[1,2]{Emanuel Tewolde$^*$}
\author[1,2,3]{\\ Vincent Conitzer}
\author[4]{Mingyu Guo}
\author[1]{Marijn Heule}
\author[5]{Lirong Xia}

\affil[1]{Carnegie Mellon University}
\affil[2]{Foundations of Cooperative AI Lab (FOCAL)}
\affil[3]{University of Oxford}
\affil[4]{University of Adelaide}
\affil[5]{Rutgers University}
\affil[ ]{\texttt{\{rberker, etewolde, conitzer, mheule\}@cs.cmu.edu, mingyu.guo@adelaide.edu.au, lirong.xia@rutgers.edu}}

\begin{document}

\maketitle

\begin{abstract}
    Core stability is a natural and well-studied notion for group fairness in multi-winner voting, where the task is to select a committee from a pool of candidates. We study the setting where voters either approve or disapprove of each candidate; here, it remains a major open problem whether a core-stable committee always exists. In this work, we develop an approach based on mixed-integer linear programming for deciding whether and when core-stable committees are guaranteed to exist. In contrast to SAT-based approaches popular in computational social choice, our method can produce proofs for a specific number of candidates \emph{independent} of the number of voters. In addition to these computational gains, our program lends itself to a novel duality-based reformulation of the core stability problem, from which we obtain new existence results in special cases. Further, we use our framework to reveal previously unknown relationships between core stability and other desirable properties, such as notions of priceability.
\end{abstract}

\section{Introduction}

Elected committees enable efficient implementation of high-stakes decisions affecting large groups of voters, forming the backbone of representative democracies. 
The same formalization of electing a ``committee'' also applies to selecting sets of items of other types---a set of locations at which to build facilities, times at which to schedule webinars, \emph{etc}. The benefits of committee elections, however, crucially rely on the committee itself being (proportionally) representative, reflecting the diverse preferences of the voters. One possible way of expressing such preferences is that of \emph{approval sets}, where each voter specifies which candidates they approve of. 
Approval preferences are simple to elicit; they also come with a range of mathematical properties that enable us to rigorously argue about the quality of an elected committee. Accordingly, \citet{Aziz17:Justified}
defined a hierarchy of \emph{representation axioms}---desiderata a committee should ideally satisfy---for approval-based multi-winner elections. Their work paved the way for a rich literature analyzing further axioms in this setting, as well as the relationships among them; \cf \citet{Lackner22:Multi} for an excellent overview. Importantly, several of these axioms guided the design of novel and efficient (multi-winner) voting rules \citep{Aziz18:Complexity,Peters20:Proportionality,Brill23:Robust,Casey25:Justified}, some of which are currently utilized in local elections in several countries \citep{Peters25:Core}.
 
Out of the hierarchy of axioms introduced by \citeauthor{Aziz17:Justified}, the one named \emph{core stability} stands out, both due to its interpretability and strength. It stands on top of this hierarchy and implies many axioms introduced later on (see, \eg, \citealp{Peters21:Proportional}). Intuitively, a committee $W$ consisting of $k$ candidates is core-stable if no coalition of voters all strictly prefer an alternative committee that they can jointly ``afford'' with their proportional share of the seats. More formally, we say that a subset of candidates $W' \subseteq C$ is a successful \emph{deviation} from $W$ if there is a subset of voters $N' \subseteq N$ such that $\frac{|W'|}{k} \leq \frac{|N'|}{|N|}$ (\ie, $N'$ can afford $W'$) and all members of $N'$ approve of strictly more members of $W'$ than $W$. Then, $W$ is core-stable if it does not admit any successful deviations. Equivalently, for any potential deviation $W'$, the support it receives from the voters (who strictly prefer it to $W$) must be insufficient to afford $W'$.

\begin{ex}[This paper's running example]
\label{ex:simple_profile}
    Consider an instance with $|N|=6$ voters and $|\alt|=5$ candidates. The approval sets of the voters are
    \begin{center}
    \begin{tabular}{ccc}
        $\vote_1 = \{c_1,c_2,c_3\}$ & $\vote_2=\{c_2,c_4\}$ & $\vote_3=\{c_2,c_4\}$ \\
       $\vote_4 = \{c_2,c_5\}$ &$\vote_5=\{c_2,c_5\}$ & $\vote_6= \{c_4,c_5\}$    
    \end{tabular}.\end{center}
    If our goal is to pick a committee with $k=3$ candidates, then $W=\{c_1,c_3,c_5\}$ is \emph{not} core-stable: $W'=\{c_2,c_5\}$ is preferred to $W$ by each voter in $N'=\{2,3,4,5\}$, achieving $\frac{|W_1'|}{k} = \frac{2}{3} \leq \frac{4}{6} = \frac{|N'|}{|N|}$. The committee $W^*= \{c_2,c_4,c_5\}$ is core-stable, as there is no successful deviation from it.
\end{ex}

Core stability also captures the idea of ``fair taxation'' \citep{Munagala21:Approximate}: if the number of seats is seen as the total resources of the community (where each agent brings in the same amount), no subgroup of voters are better off leaving the group with their share of the resources.
Importantly, the appeal of this property extends beyond political elections; it can also be a powerful guiding principle for achieving group fairness in AI/ML. For example, \citet{Chaudhury22:Fairness,Chaudhury24:Fair} apply core stability (and an adaptation thereof) to define fair outcomes in federated learning, where a set of agents (\emph{clients}) use their heterogeneous and decentralized data to train a single model. Indeed, if a subset of clients can use their own resources to train a different model that makes all of them happier, they will stop contributing their data and compute to the larger coalition. Approval-based multi-winner elections can also capture extensions of this setting, such as federated multi-objective optimization (introduced by \citealp{Yang24:Federated}), where each voter is interested in only a subset of the objectives. In the AI alignment literature, \citet{Conitzer24:Social} emphasize that core-like deviations can guide decisions to create multiple AI systems serving different subgroups, rather than a single overall system.

Despite its wide applications and strong guarantees, core stability remains highly elusive: all known multi-winner voting rules fail it, and it remains a major open problem whether a core-stable committee always exists in approval-based elections \citep{Aziz17:Justified,Cheng19:Group,Munagala22:Auditing}. While approximations \citep{Jiang20:Approximately,Peters20:Proportionality} and relaxations \citep{Cheng19:Group,Brill20:Approval} 
have been studied, until very recently the only known (unrestricted) existence result was for committees of size up to $k=3$ \citep{Cheng19:Group}. In concurrent and soon-to-appear work, \citet{Peters25:Core} pushes this boundary up to $k=8$ by showing that a specific voting rule, Proportional Approval Voting (PAV), gives core-stable committees up to this point. However, he also gives a counterexample for PAV for $k=9$, showing that this specific method cannot be pushed further. In this paper, we take an approach \textbf{free of assumptions about the voting rule used}; that is, we search over all possible voter preferences to find those where core stability comes closest to being failed by all committees. At the heart of many of our results lie techniques borrowed from \emph{automated reasoning}.

\subsection{Automated Reasoning in Social Choice}

The automated reasoning paradigm involves the development of theoretical results and interpretable proofs through the intuition gained from computer-generated proofs for  instances to which the techniques scale. \citet{Tang08:A,Tang09:Computer}  demonstrated its effectiveness in social choice by rederiving the celebrated impossibility theorems by \citet{Arrow63:Social} and Gibbard–Satterthwaite \citep{Gibbard73:Manipulation,Satterthwaite75:Strategy}. The central and recurring idea since then has been to express the input parameters to social choice problems (votes, candidate selection, \emph{etc}.) via Boolean variables, and to encode axioms via Boolean formulas over these variables. In the following years, this paradigm---and specifically SAT solving \citep{Biere21:Handbook}---has lead to advancements in a plethora of social choice settings \citep{Geist17:Computer}. Examples include problems such as ranking sets of objects \citep{Geist11:Automated}, irresolute voting rules and tournament solutions \citep{Brandt16:Finding}, fair division \citep{Brandl21:Distribution}, the no-show paradox \citep{Muolin88:Condorcet,Brandt17:Optimal,Brandl19:Strategic}, and matching markets \citep{Endriss20:Analysis}. Closer to our setting of multi-winner elections, \citet{Peters18:Proportionality} showed that forms of proportionality and strategyproofness can be incompatible with regards to some resolute voting rules.

A drawback to the SAT approach is that solvers do not scale well with the number of voters, which can get large. This is a bottleneck in our setting especially because multi-winner voting instances are parametrized by three values (number of voters, candidates, and committee seats), which further complicates the key step of extending computer proofs for small instances to the general setting. In this work, we analyze core stability for vote \emph{distributions}, which enables us to leverage ``linearity'' properties of the core \citep{Xia24:Linear} to eliminate dependencies on the number of voters. To do so, we abandon SAT methods in favor of \emph{mixed-integer linear programming}. Related approaches have previously found fruitful applications in probabilistic social choice \citep{Mennle16:Pareto,Brandl18:Proving}. They have also long
been fundamental to the adjacent field of \emph{automated mechanism design}, where (mixed integer) linear programs can be used to solve the discretized version of the general problem~\cite{Conitzer02:Mechanism,Conitzer04:Self} and have helped in proving new (im)possibility results, for example in the context of redistribution mechanisms \citep{Guo07:Worst,Guo10:Computationally,Damle24:Designing}. 
Last but not least, \citet{Peters25:Core} analyzes the PAV rule in our problem setting by describing the PAV solutions as a linear program. Overall, however, we find that automated reasoning approaches beyond SAT---especially those taking the perspective of linear theories---have been relatively underexplored in the literature. Indeed, we suspect that a variety of other social choice settings can benefit from them.

\subsection{Our Contributions}

Leveraging the formulation of core stability in terms of vote distributions, we introduce and study the problem of finding the vote distributions in which committees are ``least'' core-stable. In \Cref{sec:milp}, we cast this problem first as a nested optimization problem, and then as a mixed-integer linear program, where the optimal value corresponds to the (positive or negative) excess support the best deviation of each committee is guaranteed to have. By running experiments with our implementation of the latter program in \Cref{sec:experiments}, we identify a pattern in the optimal values of instances with few candidates, leading us to a surprising connection to the stronger axiom of Droop core stability. We show that the identified pattern indeed forms a lower bound for all instances. In \Cref{sec:dual}, we then use linear program duality to prove matching upper bounds in specific cases with small deviations or large committees, proving novel non-emptiness results and rediscovering previously known results as corollaries. In \Cref{sec:lindahl}, we use a modification of our program to identify the minimal instances for which (Droop) core stability does not imply priceability axioms from the literature. Our findings resolve previously open problems, including disproving a conjecture on Lindahl priceability. The code for our experiments is available as an open-source GitHub repository,\footnote{https://github.com/emanueltewolde/Core-MILP} and omitted proofs can be found in the appendix.

\section{Preliminaries}
\label{sec:prelims}

In \emph{approval-based multi-winner elections}, we have a set of \emph{candidates} $\alt = \{c_1, c_2, \ldots, c_m\}$, from which we have to select a committee $W \subset \alt$ of size $k \in \N$, also called a $k$-committee. We assume $0<k<m$. Denote the set of all $k$-committees and non-empty committees of size up to $k$ by 
\begin{center}
    $\calM_k \coloneq \{ W \subset \alt : |W| = k\}$, and \\ 
    $\calM_{\leq k} \coloneq \{ W' \subset \alt : 1\leq |W'| \leq k\}$
\end{center}
respectively. We refer to each $W' \in \calM_{\leq k}$ as a (potential) \emph{deviation}. Let $N \coloneq\{1,2,\ldots,n\}$ be the set of \emph{voters}, where $n \in \N$. Each voter $i \in N$ has a subset $\vote_i \subseteq \alt$ of candidates they approve of, also called an \emph{approval set}. Together, these sets form an (approval) \emph{profile} $\profile \coloneq ( \vote_i )_{i \in N}$.

We call a $k$-committee $W$ \emph{core-stable} if, intuitively, no deviation $W'$ is strictly preferred to $W$ by a subset of voters whose proportional share of candidates exceeds $|W'|$.

\begin{defn}[Core Stability]
    A $k$-committee $W \in \calM_k$ is \emph{core-stable} \wrt profile $\profile$ if for all $W' \in \calM_{\leq k}$, we have
    \begin{align}
    \label{eq:core stability}      
    |\{i \in N: |\vote_i \cap W'| > |\vote_i \cap W| \}| \cdot \frac{k}{n} < |W'|.
    \end{align}
\end{defn}

Here, $\frac{k}{n}$ is the fraction of a committee seat that a single voter can ``buy'', or dually, $\frac{n}{k}$ is the ``cost'' (expressed in number of voters) of securing a seat. The \emph{core} of a profile is defined as the set of all of its core-stable committees.

\section{Core as a Mixed-Integer Linear Program}
\label{sec:milp}

In this section, we develop a mixed-integer linear program for computing profiles for which core stability is ``least satisfiable,'' \ie, the core, if not empty, is closest to being empty. Recently, \citet{Xia24:Linear} has noted the following ``linearity'' property: The exact number of voters $n$ is not crucial to the core of a  profile $\profile$, but rather the \emph{frequency} with which each vote $\vote \subseteq \alt$ appears in $\profile$. In particular, if we multiply each voter by a constant number, then the core remains unchanged. This can be seen from \Cref{eq:core stability}; note that $k$ stays fixed. Hence, for any profile $\profile = ( \vote_i )_{i \in N}$, we can study its associated \emph{vote distribution} $\bfx \in \Delta(2^{\alt}) \cap \Q^{2^C}$, defined as $\bfx[\vote] \coloneq \frac{1}{n} \cdot |\{i \in N : \vote_i = \vote \}|$. Here, $\Delta(S) \subseteq \R^{S}$ denotes the probability simplex over a discrete set $S$. For the profile from \Cref{ex:simple_profile}, the associated vote distribution is
\begin{align}\label{eq:ex-profile}
    \bfx[\vote]= \begin{cases}
        1/3 & \text {if }\vote=\{c_2,c_4\}\text{ or }\vote=\{c_2,c_5\}\\
        1/6 & \text{if }\vote=\{c_1,c_2,c_3\}\text{ or }\vote=\{c_4,c_5\}\\
        0&\text{otherwise}
    \end{cases}.
\end{align}
Define the binary vector $\impr_{W, W'} \in \{0,1\}^{2^\alt}$ for each $W \in \calM_k$ and $W' \in \calM_{\leq k}$ to have entry $1$ at index $\vote \subseteq \alt$ if and only if $|\vote \cap W'| > |\vote \cap W|$. That is, $\impr_{W, W'}$ indicates which votes $\vote$  strictly prefer $W'$ over $W$. 

\begin{lemma}[\citealp{Xia24:Linear}]
\label{lemma:core linear}
    Given $m,k,n,$ and profile $\profile$, a $k$-committee $W$ is core-stable if and only if
    $\impr_{W,W'}^T \bfx - \frac{|W'|}{k} < 0$ for $\profile$'s vote distribution $\bfx$ and for all $W' \in \calM_{\leq k}$.
\end{lemma}

\Cref{lemma:core linear} allows us to work in the vote distribution space $\Delta(2^{\alt}) \cap \Q^{2^C}$ instead of in profile spaces $(2^\alt)^n$ for varying numbers $n \in \N$. That is because if we found a vote distribution $\bfx \in \Delta(2^{\alt}) \cap \Q^{2^C}$ with an empty core, we can construct a profile $\profile$ out of it with that same property by rescaling $\bfx$ to an element of $\N^{2^C}$, and interpret its entries as the number of voters in $N$ with the respective approval sets.

Therefore, we can decide for a given $m$ and $k$ whether there exists a profile $\profile = ( \vote_i )_{i \in N}$ with \emph{any} number of voters $n$ such that the core is empty by deciding whether 
\begin{center}
    $\exists \bfx \in \Delta(2^{\alt}) \cap \Q^{2^C} \, \forall W \in \calM_k \, \exists W' \in \calM_{\leq k} :$ 
    \\
    $\impr_{W,W'}^T \bfx - \frac{|W'|}{k} \geq 0$.
\end{center}

We  reformulate this into an optimization problem.

\begin{restatable}{prop}{maxminmax}
\label{prop:maxminmax}
    Suppose $m$ and $k$ are given. Then there exists a set of voters $N$ and profile $\profile = ( \vote_i )_{i \in N}$ for which the core is empty if and only if the optimization problem
\[    \tag{\MMM}
    \label{eq:core maxminmax}
        \max_{\bfx \in \Delta(2^{\alt}) \cap \Q^{2^C}} \min_{W \in \calM_k} \max_{W' \in \calM_{\leq k}} \impr_{W,W'}^T \bfx - \frac{|W'|}{k}
\]    has nonnegative value. 
\end{restatable}

Going beyond the question of nonnegativity, we can give the optimization in \MMM{} an interpretation \emph{\`a la} {``least core''} from cooperative game theory \citep{Shapley66:Quasi,Maschler79:Geometric}. For a fixed vote distribution $\bfx$, we define a $k$-committee $W$ to be in $\bfx$'s \emph{sizewise-least core} if it solves the inner min-max problem of \MMM{} for $\bfx$. We call the corresponding optimal objective $\mu_\bfx^*$ \emph{the value} of $\bfx$'s sizewise-least core; hence, \MMM{} essentially searches for a point $\bfx^*$ with the largest such value. Intuitively, any committee $W$ in the sizewise-least core of $\bfx$ gives rise to at most $\mu_\bfx^*$ excess voter support to any deviation $W'$ beyond what is required afford it. 
No matter the instance, the sizewise-least core always exists. Moreover, if $\mu_\bfx^*$ is negative, the sizewise-least core will be a subset of the core, making the core non-empty. In such cases, we intuitively expect committees in the sizewise-least core to be the most robust among core-stable committees against fluctuations in the vote distribution; \cf \citet{Li15:Cooperative} for an analogous study in cooperative game theory. Indeed, our interpretations above cannot be captured by notions of approximate core in prior literature (see, \emph{e.g.}, \citealp[Definition 4.11]{Lackner22:Multi}). 

In order to make the max-min-max problem \MMM{} amenable to empirical as well as theoretical analysis, we next reformulate it to a single-level optimization problem. We have tried multiple reformulations and found that a particular one yields a mixed-integer linear problem that is most efficient in practice, which we present below.
\begin{align}
 &\max_{\bfx \in \R^{2^\alt}, \, \mu \in \R, \, \bfy \in \{0,1\}^{\calM_k \times \calM_{\leq k}}} \quad  \mu \quad  \tag{\MILP}
 \\
 &\textnormal{s.t.} \, \, \sum_{\vote \subseteq \alt} \bfx[\vote] = 1 \quad \text{and} \quad \forall \vote \subseteq \alt : \bfx[\vote] \geq 0 \nonumber
\\
&\quad \, \forall W \in \calM_k : \sum_{W' \in \calM_{\leq k}} \bfy[W,W'] \geq 1 \nonumber
\\
&\quad \, \forall W \in \calM_k, W' \in \calM_{\leq k} :  \nonumber
\\
&\quad \quad \quad \mu \leq \impr_{W,W'}^T \bfx - \frac{|W'|}{k} + 3(1-\bfy[W,W']) \nonumber
\end{align}

The main idea of \MILP{} is to introduce a binary variable $\bfy[W,W']$ that evaluates to $1$ only if deviation $W'$ maximizes the inner max problem of \MMM{} for $k$-committee $W$. Despite searching over real variables, \MILP{} will admit a rational optimal solution since all of its coefficients are rational. We next show these solutions match those of \MMM.

\begin{restatable}{thm}{milp}
\label{thm:milp_equivalence}
    For any $m$ and $k$, the optimal values of \MMM{} and \MILP{} coincide, and solutions $\bfx^*$ to \MMM{} are exactly the $\bfx$-components of solutions $(\bfx^*, \mu^*, \bfy^*)$ to \MILP{}. Further, solution component $\mu^*$ represents the largest value achievable for the sizewise-least core for $m$ and $k$.
\end{restatable}

Intuitively, we replaced the inner min-max in \MMM{} with a continuous variable $\mu$. To ensure that this is a correct interpretation of $\mu$ given the chosen values of $\bfx$, we have to make sure that for every committee $W$, there in fact exists a $W'$ such that $\mu$ is at most $\impr_{W,W'}^T \bfx - \frac{|W'|}{k}$.  (This will ensure that $\mu$ will not be set too high; we do not need to worry that $\mu$ will be set too low, because the solver will try to maximize $\mu$.) To do so, we force the solver to set $\bfy[W,W']$ to $1$ for at least one $W'$, and for that one, the value of $\impr_{W,W'}^T \bfx - \frac{|W'|}{k}$ must be large enough because the last term in the last \MILP{} constraint disappears when $\bfy[W,W']=1$. (Whereas if $\bfy[W,W']=0$,
that  constraint automatically holds due to the slack of $3$ added by this last term. Any slack of $\geq 2$ would have sufficed here.)

\section{From Experiments to Lower Bounds}
\label{sec:experiments}

We implement \MILP{} using \Gurobi{} \citep{gurobi}, a popular commercial solver for mixed-integer, linear, and nonlinear optimization that guarantees global optimality (up to a small tolerance error) upon termination. All experiments in this paper were run up to the default tolerance of $10^{-4}$. \Cref{tab:milpvals} depicts the optimal values of \MILP{} that we computed for various $(m,k)$ pairs with $1 \leq k < m$. 

Our experiments show that for $m \leq 7$ candidates, there will always be a core-stable committee, for any committee size $k$, any number of votes $n$, and any vote profile $\profile$. It improves on \citet{Cheng19:Group}, which confirms experimentally that the core is non-empty for $m+n \leq 14$ and $k < m$. Concurrent and soon to appear work by \citet{Peters25:Core} subsumes our experimental (non-emptiness) results since it shows that the core is always non-empty if $m \leq 15$, or, alternatively, if $k \leq 8$. \citeauthor{Peters25:Core} obtains these results by focusing on when the PAV rule and modifications thereof are guaranteed to select a core-stable committee. For any larger $m$ or $k$, \citet{Peters25:Core} gives PAV failure modes, showing the limitations to his PAV-focused approach. In contrast, our \MILP{} resolves (non-)emptiness of core for given values of $m$ and $k$ in general, rather than focusing just on whether one particular voting rule selects core-stable committees. 

The limitations of our approach, however, lie in bounded computational resources. \MILP{} contains $2^m + 1$ continuous variables and $\binom{m}{k} \cdot \sum_{l=1}^k \binom{m}{l}$ binary variables (or a few less after eliminating certain $(W,W')$ pairs, such as when $W' \subseteq W$, without loss of optimality). For $k \approx \nicefrac{m}{2}$, these values grow super-polynomially in $m$. For small values of $m$, our implementation is still quite fast: on 8 cores, the instance $m=7$ and $k = 3$ takes $\leq 2.5$h time to run on our laptop, while \citet{Peters25:Core} briefly describes an implementation attempt for the same general problem that does not terminate within a time limit of 37h. If we investigate $m=8$, however, then the solver does not terminate within the timelimit of $72$h when $k=4$. More generally, we find the case $k \approx \nicefrac{m}{2}$ takes the longest to solve in our implementations. 

\begin{table}[t]
\centering\small
{
\renewcommand{\arraystretch}{1.5}
\begin{tabular}{|c||c|c|c|c||c|}
\hline
\textbf{k\textbackslash m} & \textbf{4} & \textbf{5} & \textbf{6} & \textbf{7} & \textbf{Fraction} \\
\hline
\hline
\textbf{1} & -0.5000 & -0.5000& -0.5000 &  -0.5000 & $= -\frac{1}{2}$ \\
\hline
\textbf{2} & -0.1667  & -0.1667 & -0.1667 &  -0.1667 & $= -\frac{1}{6}$ \\
\hline
\textbf{3} & -0.0833 & -0.0833 & -0.0833 &  -0.0833 & $= -\frac{1}{12}$ \\
\hline
\textbf{4} & --- & -0.0500 & -0.0500 &  -0.0500 & $= -\frac{1}{20}$ \\
\hline
\textbf{5} & --- & --- & -0.0333 &  -0.0333 & $= -\frac{1}{30}$ \\
\hline
\textbf{6} & --- & --- & --- & -0.0238 &  $= -\frac{1}{42}$ \\
\hline
\textbf{k} & --- & --- & --- & --- &$\frac{-1}{k(k+1)}$ \\
\hline
\end{tabular}
}
\caption{Optimal values of \MILP{} computed for various combinations of $m,k$. Columns $m \leq 3$ are omitted for brevity.}
\label{tab:milpvals}
\end{table}

That being said, our experiments continue to offer insights beyond just non-emptiness: we observe that the optimal values in \Cref{tab:milpvals} have a structure to them. They are independent of $m$ and take on the value $\frac{-1}{k(k+1)}$. If this were true for all $(m,k)$, the core would always be non-empty. While we cannot prove this in full generality in this paper, we can make progress towards such a result. First, we prove $\frac{-1}{k(k+1)}$ is a lower bound to the optimal value of \MILP{} \emph{for all} $m$ and $k$. 

\begin{restatable}{thm}{milplower} \label{thm:milp_lowerbound}
    For all $m,k$, $\MILP{} \geq \frac{-1}{k(k+1)}$. 
\end{restatable}

\begin{proof}
    Since \MILP{} is a maximization, we can prove this lower bound by giving a feasible variable assignment to \MILP{} that achieves objective value $\frac{-1}{k(k+1)}$. Given pair $k<m$, fix a subset of candidates $B \subseteq \alt$ with $|B|=k+1$. For each $\vote \subseteq \alt$, set $\bfx[\vote] \coloneq \begin{cases}
        \frac{1}{k+1} &\text{if }\vote=\{c\}\text{ with }c\in B \\ 0 &\text{otherwise} 
    \end{cases}.$ For each $W \in \calM_{k}$, fix some $c_W \in B \setminus W$ (such a $c_W$ exists as $|B|>|W|$). Set $\bfy[W,W'] \coloneq \begin{cases}
        1 &\text{if }W'=\{c_W\}\\
        0 &\text{otherwise}
    \end{cases}.$ Finally, set $\mu = \frac{-1}{k(k+1)}$, which is also equal to the objective. It is straightforward to check that this assignment is feasible for \MILP{}. 
    In particular, for any $W$ and for $W'=\{c_W\}$, we have $\impr_{W,W'}^T \bfx = \bfx[\{c_{W}\}]$, since $|W \cap \{c_W\}| =0$. Further,
    \begin{center}
        $\impr_{W,W'}^T \bfx - \frac{|W'|}{k} + 3(1-\bfy[W,W']) =  \bfx[\{c_{W}\}] - \frac{1}{k}$
        \\ 
        $= \frac{1}{k+1} -\frac{1}{k}=\frac{-1}{k(k+1)}= \mu$,
    \end{center}
    indicating that for the above $\bfx$ and $\bfy$ values, the $\mu$ is chosen maximal while staying feasible.
\end{proof}

\Cref{thm:milp_lowerbound} shows that no profile can have a sizewise-least core with value less than $\frac{-1}{k(k+1)}$. In the proof, this value emerges after subtracting $\frac{1}{k}$ (the cost of a deviation of size 1) from $\frac{1}{k+1}$ (the maximum amount of support any such deviation will get in our assignment). The former value originates from the intuition that a coalition is entitled to make decisions about a seat in the committee only if they comprise at least $\frac{1}{k}$ fraction of the total electorate, a quantity also known as the \emph{Hare quota}. The latter value, on the other hand, is reminiscent of an alternative intuition, requiring the coalition to be \emph{strictly} greater than a fraction of $\frac{1}{k+1}$, also known as the \emph{Droop quota}. Defining core stability with respect to the latter leads to a strictly stronger criterion \citep{Brill20:Approval}. The next definition introduces this notion in our framework. 
\begin{defn}[Droop core]\label{def:droopcore}  Given $m$ and $k$, and a profile $\profile$, a  $k$-committee $W \in \calM_k$ is \emph{Droop core-stable} if 
$\impr_{W,W'}^T \bfx - \frac{|W'|}{k+1} \leq 0$  for $\profile$'s vote distribution $\bfx$ and for all $W' \in \calM_{\leq k}$.
\end{defn}

\MILP{} can be easily edited to search for profiles with an empty Droop core (again for given $m,k$) by replacing each occurrence of $k$ in a denominator with $k+1$. If we call this analogous program \Droop{}, then an empty Droop core is found if its objective value is \emph{strictly larger} than 0. The proof of \Cref{thm:milp_lowerbound} then implies for any $m,k$, there is always a profile sitting at the boundary of an empty Droop core.

\begin{restatable}{cor}{droop}\label{cor:droop-core}
    For all $m,k$, \Droop{} $\geq 0$.
\end{restatable}

Indeed, when we run \Droop{} for the values of $m,k$ in \Cref{tab:milpvals}, the program converges to a value of $0$, showing that the Droop core is always non-empty for these values.\footnote{Relatedly, PAV satisfies Droop core for $k \leq 5$ \citep{Peters25:Core}.} Further, \Cref{cor:droop-core} indicates that Droop quota is the ``best we can hope for'' if we require non-emptiness, since for any $m$ and $k$, core notions with any smaller quota (giving deviating coalitions more power) will be empty for some vote distribution $\bfx$. A similar observation was previously made by \citet[Remark 3.6]{Jason18:Thresholds} for related fairness properties. 

Next, we study when the lower bounds of \Cref{thm:milp_lowerbound} and \Cref{cor:droop-core} can be matched with an identical upper bound. 

\section{Upper Bounds Using Duality}\label{sec:dual}
    We derive a general strategy for proving upper bounds on $\MILP{}$, and prove an upper bound of $\frac{-1}{k(k+1)}$ for certain special cases. More generally, our approach provides a novel way of proving core non-emptiness results.
    
    The key observation in our approach is that for any fixed values of the integer variables $\bfy$ in \MILP{}, we get a linear program over variables $\mu$ and $\bfx$, the optimal value of which is the largest value \MILP{} can achieve using that $\bfy$. Without loss of optimality of \MILP{}, we can restrict our attention to instances of $\bfy$ that have $\bfy[W,W']=1$ for only one deviation $W' \in \calM_{\leq k}$ for each committee $W \in \calM_k$ (call this deviation $D_W$). Each such $\bfy$ then corresponds to a \emph{deviation function} $D: \calM_{k} \rightarrow \calM_{\leq k}$ such that $D(W)= D_W$ for each $W \in \calM_k$. Taking the dual of the linear program associated with $D$ (and performing some simplification steps to significantly decrease the number of variables; see the proof of \Cref{thm:dual_eq}), we obtain the following linear program.
     \begin{align}
        &\min_{\bfq \in \R^{\calM_k}, u \in \R} \quad  u - \sum_{W \in \calM_k} \frac{|D(W)|}{k}\bfq[W] \tag{\DLP} 
         \\
        &\textnormal{s.t.} \, \, \sum_{W \in \calM_k} \bfq[W] = 1 \, \, \text{and} \, \, \forall W \in \calM_k : \bfq[W] \geq 0 \nonumber
        \\
        &\quad \, \forall \vote \subseteq \alt : \sum_{\substack{{W \in \calM_k:} \\{|D(W) \cap \vote| > |W \cap \vote| }}}  \bfq[W] \quad \leq u \nonumber
    \end{align}
    
    As with \MILP{} and \Droop{}, we refer to the analogous linear program for Droop core (where the $k$ in the objective is replaced with $k+1$) as \DrDLP. Since the optimal value of \MILP{} is the optimal value over all assignments of $\bfy$, we can use strong duality and obtain an upper bound on \MILP{} by bounding \DLP{} across all deviation functions.
    \begin{restatable}{thm}{dual} \label{thm:dual_eq}
        Given $m,k,$ and value $ v\in \R$, we have $\MILP{} \leq v$ (resp.\ $\Droop{} \leq v$) if and only if $\DLP{} \leq v$ (resp.\ $\DrDLP \leq v$) for all functions $D: \calM_k \rightarrow \calM_{\leq k}$.
    \end{restatable}
    
    Furthermore, a closer look at \DLP{} shows that $\bfq$ defines a probability distribution over $\calM_k$. Combining this observation with \Cref{lemma:core linear} and \Cref{thm:dual_eq} lends itself to a novel way of formulating the core non-emptiness problem.

    \begin{restatable}{cor}{lottery}\label{cor:prob}
        Given $m$ and $k$, the core (resp.\ Droop core) is non-empty for all profiles / vote distributions if and only if the following statement is true:
        
        \noindent\fbox{\begin{minipage}{\columnwidth-0.7em}      
        For every function $D: \calM_k \rightarrow \calM_{\leq k}$, there exists a distribution $\bfq \in \Delta(\calM_k)$ s.t.~for all votes $\vote \subseteq \alt$, we have
        \begin{center}
            $\displaystyle \underset{W \sim \bfq}{\Prob}\big[|W \cap \vote| < |D(W) \cap \vote| \big] \underset{(\textnormal{resp.\ } \leq)}{<} \, \frac{ \underset{W \sim \bfq}{\Mean}\big[\left|D(W)\right|\big]}{k \, (\textnormal{resp.\ } k+1)}$
        \end{center}
        \end{minipage}}
    \end{restatable}
    
   The boxed statement in \Cref{cor:prob} is notably similar to a lemma proven by \citet[Inequality (3)]{Cheng19:Group} which implies the existence of \emph{stable lotteries} (distributions over $\calM_k$ such that the expected support of any deviation is not sufficient), a weaker result than core non-emptiness. In \Cref{appsec:connection}, we reinterpret their formulation (where $D$ is replaced by distributions over $\calM_{\leq k}$) in our framework. We believe \Cref{cor:prob} lays the foundations for a probabilistic analysis towards proving core non-emptiness, possibly similar to \citeauthor{Cheng19:Group}'s approach for stable lotteries. 
    
     Next, we illustrate the strength of \Cref{thm:dual_eq} by proving upper bounds to \MILP{} in special cases. Our results improve on previously known results, implying them as corollaries.

    \paragraph{Small deviations} The example profile we used in the proof of \Cref{thm:milp_lowerbound} relied on singleton deviations ($|D(W)|=1$ for all $W$). Indeed, in all of the the experiments in \Cref{tab:milpvals}, we see that the objective-maximizing values of $\bfy$ correspond to singleton deviations. We now formalize this intuition by showing that our lower bound for \MILP{} is tight for singleton deviations, and the objective can only get worse when we add a single non-singleton deviation.

    \begin{thm}\label{thm:singleton}
        Say we are given $m,k,$ and a deviation function $D$ such that $|D(W)|=1$ for all $W \in \calM_{k}$ but possibly one $W^*$. Then, $\DLP{} \leq -\frac{1}{k(k+2-|D(W^*)|)}$ (resp.\ $\DrDLP{} \leq 0$).
    \end{thm}
    \begin{proof}
        Start with $W_1 \coloneq W^*$, and say $t \coloneq |D(W^*)|$. For $i = \{2,\ldots,k+2-t\}$, pick $W_i$ to be some arbitrary committee in $\calM_k$ such that $\bigcup_{j \in [i-1]} D(W_j) \subseteq W_i$  (the union will contain at most $i-2+t \leq k$ elements). Set $\bfq \in  \Delta(\calM_k)$ as
        $\bfq[W] = \begin{cases}
                \frac{1}{k+2-t} &\text{if }W=W_i\text{ for some }i \in [k+2-t]\\
                0 &\text{otherwise}
            \end{cases}.$
        Fix any $\vote \subseteq \alt$. We claim that the inequality $|\vote \cap W_i| < |\vote \cap D(W_i)|$ can be true for at most one $ i\in [k+2-t]$. Say $i^*$ is the smallest $i\in [k+2-t]$ such that $\vote \cap D(W_i) \neq \emptyset$ (if no such $i^*$ exists, we are done). For any $j<i^*$, we have $|\vote \cap D(W_j)|=0$, so the inequality must be false. For any $j> i^*$, we have $D(W_{i^*}) \subseteq W_j$; therefore, $|\vote \cap W_{j}| \geq |\vote \cap D(W_{i^*})| \geq 1 = |D(W_j)| \geq  |\vote \cap D(W_j)|$, so the inequality must be once again false. Hence, the inequality can only be true for $i^*$, implying $u=\frac{1}{k+2-t}$ with the above $\bfq$ is a feasible assignment to \DLP{} (and \DrDLP). Further, we have $\sum_{W \in \calM_k} |D(W)| \cdot \bfq[W]=\frac{t+(k+1-t)}{k+2-t}=\frac{k+1}{k+2-t}$. Dividing by $k$ (resp.\ $k+1$) and subtracting from $u$ gives the desired upper bound for \DLP{} (resp.\ \DrDLP).
    \end{proof}

    \Cref{thm:singleton} has several implications. First, for $t=1$, it shows that the lower bound for \MILP{} is exactly met when the deviations are restricted to singletons; further, when we add one non-singleton deviation, we only get \emph{farther} from core emptiness. It also opens up a novel possibility for proving core non-emptiness in general, namely, by proving a similar effect from adding further non-singleton deviations to $D$. 
    Second, \Cref{thm:singleton} implies a known result as a corollary: a weaker version of the core named \emph{justified representation} (where deviations are restricted to $\calM_1$ rather than $\calM_{\leq k}$) can always be satisfied, regardless of whether we use Hare or Droop quota; indeed, the voting rule PAV is known to satisfy it with either quota \citep{Aziz17:Justified,Jason18:Thresholds}.
    
    Adapting our proof of \Cref{thm:singleton} to broader classes of deviation functions presents several challenges.\footnote{For example, say $|D(W)|=2$ for all $W \in \calM_k$, rather than singletons. Using the same construction as above, where for each $i\geq 1 $ we force $\cup_{j<i }D(W_j) \subseteq W_i$, we can only pick $\approx \frac{k}{2}$ committees before the union has more elements than a single committee can contain. Further, there may be an $\vote \subseteq \alt$ with $|W_i \cap \vote| < |D(W_i) \cap \vote|$ for multiple $i$ (\emph{e.g.}, since $|D(W_1) \cap \vote| = 1$ and $|D(W_2) \cap \vote| = 2$). The left-hand side of \Cref{cor:prob} will then be $\approx \frac{2}{k/2}$ for the uniform distribution over $\{W_i\}$, which violates the inequality as the right-hand side is $\frac{2}{k}$.} Nonetheless, an analogous construction can lead us to proving novel non-emptiness results, as we show next.

    \paragraph{Large committees} We turn to the setting of $m=k+1$, for any $k$. In words, the problem is to select a single candidate that will not be in the committee. Unlike the opposite extreme of $k=1$ (in which case any candidate that is approved by some voter is a core-stable 1-committee, and picking the candidate approved by the most voters is sufficient for Droop core), the core is not trivially non-empty when $m=k+1$, as the next example shows.
    \begin{ex}\label{ex:large-core}
        Consider the profile from \Cref{ex:simple_profile}, the vote distribution of which is in \eqref{eq:ex-profile}, this time for $k=m-1=4$. Fix $W=\{c_1,c_3,c_4,c_5\}$ and $W'=\{c_2,c_4,c_5\}$. We have
        \begin{center}
            $\impr_{W,W'}^T \bfx =\bfx[\{c_2,c_4\}]+\bfx[\{c_2,c_5\}]= \frac{2}{3}>\frac{3}{5}= \frac{|W'|}{k+1}$,
        \end{center}
        showing $W$ is not in the Droop core. If we slightly modify the profile such that $\vote_1=\{c_2\}$ instead, then $W$ is not even in the core, as certified by the same deviation $W'$.
    \end{ex}
    Nevertheless, our next result shows that the core is always non-empty in this case, even with Droop quota. 
    \begin{restatable}{thm}{kplusone}\label{thm:m=k+1}
        For any $m$ and $k=m-1$, given any deviation function $D$, we have $\DLP \leq -\frac{1}{k(k+1)}$ (resp.\ $\DrDLP{} \leq 0$).
    \end{restatable}

    The proof follows from a construction similar to that of \Cref{thm:singleton}. However, unlike the latter, \Cref{thm:m=k+1} does not put any restrictions on the deviation function $D$. As a result, \Cref{thm:dual_eq} gives us core non-emptiness in this setting.\footnote{The concurrent work of \citet{Peters25:Core} shows that PAV satisfies core stability with $m=k+1$ (using Hare quota). Our \Cref{cor:m=k+1core} strengthens this non-emptiness result using the Droop core. As noted by \citet{Casey25:Justified}, improving an axiom satisfiability result from Hare to Droop can often be nontrivial. Still, we note that this result is also obtainable by combining an analogous argument to that of \citeauthor{Peters25:Core} using the observation that PAV still satisfies EJR+ with Droop quota \citep{Jason18:Thresholds,Brill23:Robust}.}
    \begin{cor}\label{cor:m=k+1core}
        The core is always non-empty for any $m$ and $k=m-1$, even using Droop quota.
    \end{cor}

    Overall, \Cref{thm:singleton,thm:dual_eq} demonstrate that our dual formulation offers a novel framework for deriving core non-emptiness results and investigating restricted notions of core stability in terms of deviation powers.

\section{Relationship to Other Axioms}\label{sec:lindahl}

We now modify \MILP{} to investigate the relationship of core stability with other axioms in the literature. Our experiments here resolve previously open problems by efficiently finding counterexamples. They also identify the minimal $m$ and $k$ values for which a counterexample exists in the first place.

\paragraph{Lindahl priceability}
\citet{Munagala22:Auditing} introduce \emph{Lindahl priceability} as another axiom in approval-based multi-winner elections, based on the idea of market clearing. To introduce it, we can think of each voter $i \in N$ as having a budget of $1$. Given a committee $W \in \calM_k$, any voter is able to switch to a strictly preferred alternative set of candidates $T \subseteq \alt$ (not necessarily bounded in size by $k$) if they can ``afford'' all candidates in $T$. The question of whether $W$ is Lindahl priceable becomes whether one can set prices (from voters to candidates) in a way that no candidate is cumulatively overpriced (specifically, larger than $\nicefrac{n}{k}$) and no voter can afford a strictly preferred set of candidates. 

\begin{restatable}{defn}{lindahlpbility}\label{def:lindahl}
    A $k$-committee $W$ is \emph{Lindahl priceable} with respect to profile $\profile$ if $\exists$ a price system 
    $\{\bfp[i,c]\}_{i \in N, c \in \alt} \geq 0$ such that
    \begin{enumerate}
        \item $\forall c \in \alt$: $\sum_{i \in N} \bfp[i,c] \leq \frac{n}{k}$, and
        \item $\forall i \in N, T \subseteq \alt$: $|\vote_i \cap T| > |\vote_i \cap W|  \Rightarrow \underset{c \in T}{\sum} \bfp[i,c] > 1$.
    \end{enumerate}
\end{restatable}

Lindahl priceability implies \emph{weak priceability}, which can be defined via \Cref{def:lindahl}, except we restrict the $T \subseteq \alt$ in condition 2 to be of the form $\{d\}  \cup (\vote_i \cap W)$ s.t. $d \in \vote_i \setminus W$, \ie, each voter can only \emph{add} candidates to their approved ones in $W$.\footnote{The notion that we here call weak priceability was introduced by \citet{Munagala22:Auditing}, who, in that version of their paper, stated it is equivalent to \emph{priceability} introduced by \citet{Peters20:Proportionality}. However, as we show in \Cref{appsec:pbility}, priceability is (1) strictly stronger than weak priceability and (2) incomparable with Lindahl priceability. As a result, our \Cref{thm:no-imply} also proves (Droop) core does not imply priceability.} While it is known that weak priceability is not sufficient for core stability, we have found no prior work with an explicit counterexample showing it is also not necessary. On the other hand, \citet{Munagala22:Auditing} show that Lindahl priceability also implies core stability.
In a later report,
\citet{Munagala24:Core} state ``even though Lindahl priceability implies core stability, we do not know if it is strictly
stronger than the core. We conjecture that these two notions are the same.'' In contrast, we are able to show through our framework that this conjecture is false.\footnote{We suspect this conjecture might have been an oversight on the part of the authors, as checking whether a committee fails Lindahl priceability can be done in polynomial time  \citep{Munagala22:Auditing}, whereas the equivalent problem for core stability is $\NP$-hard \citep[Thm 5.3]{Brill20:Approval}. Hence, the two axioms cannot be equivalent unless $\P=\NP$. In any case, we will see that our minimal counterexample in \Cref{thm:no-imply} disproves the conjecture unconditionally, and also shows Lindahl priceability is strictly stronger than the core \emph{and} weak priceability together (\Cref{cor:linhdahl-no-imply}).} 

\paragraph{Combining it with the core} In order to incorporate weak/Lindahl priceability into \MILP, we first show that both axioms are compatible with our framework of vote distributions, mirroring what \Cref{lemma:core linear} showed for core stability.\begin{restatable}{lemma}{linprice}\label{lemma:pbility-linear} 
    For any given $m,k,n$, and profile $\profile$, a $k$-committee $W$ is weakly (resp.\ Lindahl) priceable iff there exists a price system $\bfp = \{\bfp[\vote,c]\}_{\vote \subseteq \alt, c \in \alt} \geq 0$ such that
    \begin{enumerate}
        \item[1.] $\forall c \in \alt$: $\sum_{\vote \subseteq \alt} \bfx[\vote] \cdot \bfp[\vote,c] \leq \frac{1}{k}$, and
        \item[2.] $\forall \vote \subseteq  \alt, d \in \vote \setminus W, T = (\vote \cap W) \cup \{d\}$ : $\underset{c \in T}{\sum} \bfp[\vote,c] > 1$ \footnotesize(resp.\  $\forall A, T \subseteq \alt$ s.t.~$|\vote \cap T| > |\vote \cap W|$),
    \end{enumerate}
    where $\bfx$, again, is the vote distribution associated with $\profile$. 
\end{restatable}

Fixing $\bfx$ and $W$, this forms a system of linear (strict) inequalities over variables $\bfp$. In \Cref{appsec:dualpbility}, we describe how we leverage linear program duality once again to get a dual program of this system. If the dual 
admits a feasible solution, then it serves as a certificate to $W$ not being weakly/Lindahl priceable for $\bfx$. This allows us to search for a core-stable committee $W^*$ that is not weakly/Lindahl priceable. Namely, for any desired counterexample $W^*$, \wlogg{} $W^* = \{c_1, \dots, c_k\}$, minimize $\mu$ over vote distributions $\bfx$ such that $\mu \geq \max_{W' \in \calM_{\leq k}} \impr_{W^*,W'}^T \bfx - \frac{|W'|}{k}$. To this linear program, add constraints of the weak (resp.\ Lindahl) priceability dual program for $W^*$. This results in a program with quadratic constraints,\footnote{Recall that the weak/Lindahl priceability dual is only linear as long as we do not also optimize over the vote distribution $\bfx$.} whose optimal value is negative iff $W^*$ is core-stable but not weakly (resp.\ Lindahl) priceable. Using this approach, we show that even the stronger axiom of Droop core does not imply either priceability axiom.

\begin{restatable}{thm}{noimply}
\label{thm:no-imply}
    There exists a vote distribution for $m=5$ and $k=3$ (resp.\ $m=4$ and $k=2$) with a $k$-committee that is Droop core-stable but not weakly (resp.\ Lindahl) priceable. 
\end{restatable}

Another strength of our program is that it enables us to confirm the \emph{minimality} of these counterexamples, in the sense that it rules out counterexamples for any $(m',k')$ with $m' < m$ or $[m' = m, k' < k]$. In contrast, hand-designed counterexamples for similar problems can get quite large---such as that no welfarist rule (a class of voting rules including PAV) is priceable for $k=57$ and $m=669$ \citep{Peters20:Proportionality}---without revealing whether simpler ones exist. Further, leveraging minimality, we derive that Lindahl priceability is strictly stronger than both axioms it implies.

\begin{restatable}{cor}{corepnotlp}\label{cor:linhdahl-no-imply}
    There exists a vote distribution for $m=4$ and $k=2$ that admits a $k$-committee that is (Droop) core-stable and weakly priceable, but not Lindahl priceable.
\end{restatable}

We end this section with remarking that while we obtain the counterexamples in \Cref{thm:no-imply} via \Gurobi{} experiments, we can construct simple human-readable proofs uisng the optimal variable assignments in the dual weak/Lindahl priceability program. We demonstrate this with \Cref{ex:simple_profile}, whose vote distribution in \eqref{eq:ex-profile} we obtained through one of our programs (showing core stability does not imply weak priceability for $m=5$ and $k=3$). Here, $W=\{c_1,c_2,c_3\}$ is core-stable. 
For the sake of contradiction, assume $W$ is also weakly priceable, certified by a price system $\bfp =$ $\{\bfp[\vote,c]\}_{\vote \subseteq \alt, c \in \alt}$. Condition 1 of \Cref{lemma:pbility-linear} then implies: 
\begin{center}
    $\displaystyle \frac{\bfp[\{c_2,c_4\}, c_2]}{3} + \frac{\bfp[\{c_2,c_5\}, c_2]}{3} \leq  \sum_{\vote \subseteq \alt} \bfx[\vote]  \bfp[\vote,c_2] \leq \frac{1}{3}$
    \\
    $\displaystyle \frac{\bfp[\{c_2,c_4\}, c_4]}{3} + \frac{\bfp[\{c_4,c_5\}, c_4]}{6} \leq  \sum_{\vote \subseteq \alt} \bfx[\vote]  \bfp[\vote,c_4] \leq \frac{1}{3}$
    \\
    $\displaystyle \frac{\bfp[\{c_2,c_5\}, c_5]}{3} + \frac{\bfp[\{c_4,c_5\}, c_5]}{6} \leq  \sum_{\vote  \subseteq \alt} \bfx[\vote]  \bfp[\vote,c_5] \leq \frac{1}{3}$
\end{center}    
Similarly, condition 2 of \Cref{lemma:pbility-linear} yields
\begin{align}
    &\bfp[\{c_2,c_4\}, c_2] + \bfp[\{c_2,c_4\}, c_4] >1 \label{eq:1/3-1}\\
    &\bfp[\{c_2,c_5\}, c_2] + \bfp[\{c_2,c_5\}, c_5] >1  \label{eq:1/3-2}  \\
    &\bfp[\{c_4,c_5\}, c_4] >1 \quad \text{and} \quad \bfp[\{c_4,c_5\}, c_5]  >1 \label{eq:1/6}.
\end{align}
If we multiply \eqref{eq:1/3-1}-\eqref{eq:1/3-2} by $-1/3$, multiply \eqref{eq:1/6}
by $-1/6$, and add them up together with the previous three inequalities from Condition 1, we get the contradiction $0<0$.

\section{Conclusion \& Future Directions}
\label{sec:disc}

We used modifications to \MILP{} to explore core's logical relationship with other axioms, proving minimal \emph{incomparability} results. A natural next step is to use our program to search for \emph{incompatibility}: Does there exist a profile for which no core-stable committee satisfies (weak) priceability? This is currently an open problem. Similarly, to the best of our knowledge, whether the core is compatible with other axioms such as EJR+ and committee monotonicity is not known. These axioms can be incorporated into \MILP, searching for counterexamples of simultaneous satisfiability. Another direction is to extend our tools to settings different from (but related to) approval elections, such as ``thumbs up and down voting'', where each voter has three options per candidate: approval, neutral, and disapproval \citep{Kraiczy25:Proportionality}. Overall, we hope our framework will open up new directions for using MILPs in social choice settings.

\section*{Acknowledgements}
We thank Chase Norman, Ulle Endriss, Paul Gölz, Dominik Peters, Kangning Wang, Jan Vondr{\'a}k, Edith Elkind, and Markus Brill for helpful discussions and feedback at various stages of this project. R.E.B., E.T., and V.C.\ thank the Cooperative AI Foundation, Macroscopic Ventures (formerly Polaris Ventures / the Center for Emerging Risk Research), and Jaan Tallinn's donor-advised fund at Founders Pledge for financial support. R.E.B.\ and E.T.\ are also supported by the Cooperative AI PhD Fellowship. M.H.\ is supported by the National Science Foundation (NSF) under grant CCF-2415773. L.X.\ acknowledges NSF 2450124, 2517733, and 2518373 for support.

\bibliography{refs}

\appendix
\onecolumn 

In this appendix, we first provide some context on the setup of our experiments, and then prove the theoretical statement made in each section of the main body.

\section{Experimental Details}

Our well-documented experiment code can be found in the supplementary material. All of our experiments are conducted with the commercial solver \Gurobi{}; thus, a valid Gurobi license is required in order to run the experiments. Moreover, since \Gurobi{} is a global optimizer, we (1) do not have to choose a metric, (2) can only conduct a single experiment run per instance (parameterized by $m$ and $k$), and summaries of results (beyond average/median) are not applicable, and (3) statistical tests for judgment of improvements are not applicable either. No hyperparameter finetuning is needed.

All experiments but one were conducted on a MacBook Pro 16-inch (2021) equipped with an Apple M1 Pro processor (10-core CPU with 8 performance and 2 efficiency cores, 16-core GPU) with 16GB unified memory. The \MILP{} experiment for $m=8, k=4$, on the other hand, was conducted on a compute node equipped with dual AMD EPYC 7742 processors (64 cores each, 128 cores total, 2.25-3.40 GHz) and 256GB RAM.

\section{On \Cref{sec:milp}}

\subsection{Proof of \Cref{prop:maxminmax}}
\maxminmax*
\begin{proof}
    By \Cref{lemma:core linear}, the question of whether the core of a given profile $\profile$ is empty, that is, whether $\profile$ does not admit a core-stable committee $W$, is equivalent to the truth value of
    \[
        \forall W \in \calM_k \, \exists W' \in \calM_{\leq k} : \, \impr_{W,W'}^T \bfx - \frac{|W'|}{k} \geq 0
    \]
    for $\profile$'s corresponding vote distribution $\bfx$. This itself is equivalent to whether
    \[
        \min_{W \in \calM_k} \max_{W' \in \calM_{\leq k}} \impr_{W,W'}^T \bfx - \frac{|W'|}{k}
    \]
    has a nonnegative value. Therefore, deciding whether there exists a $N$ and $\profile = ( \vote_i )_{i \in N}$ with an empty core is equivalent to deciding whether \MMM{} $\geq 0$.
\end{proof}
\subsection{Proof of \Cref{thm:milp_equivalence}}
\milp*
\begin{proof}
    First, we observe generally that for any feasible $\bfx$ and any $W,W',$ we have
    \begin{align}
        \label{eq:bound}
        \impr_{W,W'}^T \bfx - \frac{|W'|}{k} \in [0 - 1, 1 - 0] = [-1,1] \, .
    \end{align}
    (\MMM{} $\leq$ \MILP{}): Let $\bfx$ be an optimal point of \MMM{}. We show that it forms a feasible point $(\bfx, \mu, \bfy)$ for \MILP{} defined as follows. Set $\mu := \min_{W \in \calM_k} \max_{W' \in \calM_{\leq k}} \impr_{W,W'}^T \bfx - \frac{|W'|}{k}$, and for $W \in \calM_k$, set $\bfy[W,W']$ to $1$ for any $W' \in \argmax_{W'' \in \calM_{\leq k}} \impr_{W,W''}^T \bfx - \frac{|W''|}{k}$, and $0$ otherwise. Then, the first two \MILP{} constraints are satisfied. The last constraint is vacuously satisfied for any $\bfy[W,W']=0$, because $\mu - \Big(\impr_{W,W'}^T \bfx - \frac{|W'|}{k} \Big) \leq 2$ by (\ref{eq:bound}). If $\bfy[W,W']=1$, on the other hand, the last constraint is still satisfied by the choice of $\mu$ and $\bfy$.

    \noindent(\MMM{} $\geq$ \MILP{}): Let $(\bfx, \mu, \bfy)$ be a feasible variable assignment in \MILP{}. We can assume without loss of optimality that the components of $\bfx$ are rational, since all constraints of \MILP{} use rational coefficients. We will show $\mu \leq \MMM{}$, which implies $\MILP{} \leq \MMM{}$, as $\MILP{}$ has $\mu$ as its objective value. By the constraint on $\bfy$, for each $W \in \calM_k$ there must exist some $W'_W \in \calM_{\leq k}$ such that $\bfy[W,W'_W]=1$. By the last constraint of \MILP{}, this implies $\mu \leq \impr_{W,W'_W}^T \bfx - \frac{|W'_W|}{k} + 3 \cdot 0 \leq  \max_{W'\in \calM_{\leq k}} \impr_{W,W'}^T \bfx - \frac{|W'|}{k}$. Since this is true for each $W \in \calM_k$, this implies 
    $$\mu \leq \min_{W \in \calM_k}  \max_{W'\in \calM_{\leq k}} \impr_{W,W'}^T \bfx  - \frac{|W'|}{k} \leq \max_{\bfx' \in \Delta(2^{\alt}) \cap \Q^{2^C}} \min_{W \in \calM_k} \max_{W' \in \calM_{\leq k}} \impr_{W,W'}^T \bfx' - \frac{|W'|}{k} = \MMM{}.$$
\end{proof}

\section{On \Cref{sec:experiments}}
\subsection{Proof of \Cref{cor:droop-core}}
\droop* 

\begin{proof}

Recall that \Droop{} is identical to \MILP{} except the last constraint is now $\forall W \in \calM_k, W' \in \calM_{\leq k}:$
\begin{align*}
\mu \leq \impr_{W,W'}^T \bfx - \frac{|W'|}{k+1} + 3(1-\bfy[W,W']).
\end{align*}
Consider the same variable assignment as in the proof of \Cref{thm:milp_lowerbound}, except $\mu=0$. Once again, it is straightforward to check that this assignment is feasible for \Droop{}. In particular, for any $W$ and for $W'=\{c_W\}$, we have 
    \begin{center}
        $\impr_{W,W'}^T \bfx - \frac{|W'|}{k+1} + 3(1-\bfy[W,W']) =  \bfx[\{c_{W}\}] - \frac{1}{k+1}$
        \\ 
        $= \frac{1}{k+1} -\frac{1}{k+1}=0$,
    \end{center}
    indicating that for these values of $\bfx$ and $\bfy$, the $\mu$ is chosen maximal while staying feasible.
\end{proof}

\section{On \Cref{sec:dual}}
\subsection{Proof of \Cref{thm:dual_eq}}
\dual*
 \begin{proof} We will prove the statement for \MILP, the proof for \Droop{} follows identically. For convenience, we first restate \MILP. 
\begin{align}
 &\max_{\bfx \in \R^{2^\alt}, \, \mu \in \R, \, \bfy \in \{0,1\}^{\calM_k \times \calM_{\leq k}}} \quad  \mu \quad  \tag{\MILP}
 \\
 &\textnormal{s.t.} \, \, \sum_{\vote \subseteq \alt} \bfx[\vote] = 1 \quad \text{and} \quad \forall \vote \subseteq \alt : \bfx[\vote] \geq 0 \nonumber
\\
&\quad \, \forall W \in \calM_k : \sum_{W' \in \calM_{\leq k}} \bfy[W,W'] \geq 1 \nonumber
\\
&\quad \, \forall W \in \calM_k, W' \in \calM_{\leq k} :  \nonumber
\\
&\quad \quad \quad \mu \leq \impr_{W,W'}^T \bfx - \frac{|W'|}{k} + 3(1-\bfy[W,W']) \nonumber
\end{align}
 
    For any fixed assignment to $\bfy$, \MILP{} becomes a linear program over variables $\bfx$ and $\mu$, the optimal value of which is at most that of \MILP{}; call this linear program $\LPy{\bfy}$. Define $Y=\{\bfy \in \{0,1\}^{\calM_k \times \calM_{\leq k}}: \sum_{W' \in \calM_{\leq k}} \bfy[W,W']=1 ~\forall W\in \calM_k\}$. Say $(\bfx^*, \mu^*, \bfy^*)$ is a feasible variable assignment that maximizes \MILP. Without loss of optimality, we can assume $\bfy^* \in Y$. This is because for any $\bfy$ such that $\bfy[W,W'_1]=\bfy[W,W'_2]=1$ for some $W \in \calM_k$ and two distinct $W'_1,W'_2 \in \calM_{\leq k}$, setting $\bfy[W,W'_2]=0$ will only increase the right hand side (i.e., the upper bound) of some $\leq$ constraints of \MILP, hence not hurting feasibility. Therefore, using optimality of $(\bfx^*, \mu^*, \bfy^*)$ with $\bfy^* \in Y$, we get $\MILP = \mu^* = \LPy{\bfy^*} \leq \max_{\bfy \in Y} \LPy{\bfy} \leq \MILP$, which yields the equality
    \begin{align}
        \MILP = \max_{\bfy \in Y} \LPy{\bfy}.\label{eq:duality}
    \end{align}
    Hence, we can upper bound \MILP{} by upper bounding $\LPy{\bfy}$ for all $\bfy \in Y$. To do so, for each $\bfy \in Y$, define the map $D_\bfy: \calM_{k} \rightarrow \calM_{\leq k}$ mapping each $W \in \calM_k$ to the unique $W' \in \calM_{\leq k}$ for which $\bfy[W,W']=1$. Then taking the dual of of $\LPy{\bfy}$ gives us

\begin{align}
 &\ \min_{\substack{u\in \mathbb{R}, \\ \bfq \in \mathbb{R}^{\calM_k \times \calM_{\leq k}}} } \quad  u +3\left(\sum_{\substack{{W \in \calM_k} \\{W' \in \calM_{\leq k}\setminus \{D_\bfy(W)\}}}} \bfq[W,W'] \right)  - \sum_{\substack{{W \in \calM_k} \\{W' \in \calM_{\leq k}}}} \frac{|W'|}{k}\bfq[W,W']  \quad  \tag{$\DPy{\bfy}$}
 \\
 &\textnormal{s.t.} \, \,  \sum_{\substack{{W \in \calM_k} \\{W' \in \calM_{\leq k}}}} \bfq[W,W'] = 1 \quad \text{and} \quad \forall W \in \calM_k , W' \in \calM_{\leq k}: \bfq[W,W'] \geq 0, \nonumber
\\
&\quad \, \forall \vote \subseteq \alt : u \geq \sum_{\substack{{W \in \calM_k, W' \in \calM_{\leq k}:}\\{|W'\cap \vote| > |W \cap \vote|}}}  \bfq[W,W'] \nonumber
\end{align}

    Using strong duality, we have $\DPy{\bfy}=\LPy{\bfy}$. Recall that in \MILP{}, the coefficient 3 in front of $(1-\bfy[W,W'])$ is simply chosen to provide a large slack to the constraint whenever $y[W,W']=0$. Therefore, the number 3 can be replaced by any larger number $M \geq 3$ without affecting the feasible solution set of $\LPy{\bfy}$: for any $W,W'$ such that $\bfy[W,W']=1$, the coefficient will disappear, and for for any $W,W'$ such that $\bfy[W,W']=0$, the constraint will be satisfied by all $\bfx$ and $\mu \leq 1$ (any $\mu >1$ will not be feasible for any $M$ anyway, due to the constraints for which $\bfy[W,W']=1$). This implies that additionally constraining $\DPy{\bfy}$ with $\bfq[W,W']=0$ for all $W\in \calM_k$, $W' \in \calM_{\leq k} \setminus \{D_\bfy(W)\}$ will not change its optimal value, since if all optimal variable assignments involved $\bfq[W,W']> 0$ for some $W' \neq W$ we could have replaced 3 with a larger number and increase the value of the dual without changing that of the primal, which contradicts strong duality. This allows us to significantly decrease the number of variables in the dual to get
 \begin{align}
        &\min_{\bfq \in \R^{\calM_k}, u \in \R} \quad  u - \sum_{W \in \calM_k} \frac{|D_\bfy(W)|}{k}\bfq[W] \tag{$\DLP_{\bfy}$} 
         \\
        &\textnormal{s.t.} \, \, \sum_{W \in \calM_k} \bfq[W] = 1 \, \, \text{and} \, \, \forall W \in \calM_k : \bfq[W] \geq 0 \nonumber
        \\
        &\quad \, \forall \vote \subseteq \alt : \sum_{\substack{{W \in \calM_k:} \\{|D_\bfy(W) \cap \vote| > |W \cap \vote| }}}  \bfq[W] \quad \leq u \nonumber
    \end{align}
    while still ensuring $\DLP_{\bfy}=\DPy{\bfy}=\LPy{\bfy}$. Using \eqref{eq:duality}, this implies $\MILP=\max_{\bfy\in Y} \DLP_{\bfy}$. Since there is a one-to-one correspondence between $\bfy \in Y$ and functions from $D: \calM_{k}$ to $\calM_{\leq k}$, the theorem statement follows.
 
\end{proof}

\subsection{Proof of \Cref{cor:prob}}

\lottery*

\begin{proof}
    Any feasible assignment to $\bfq$ in \DLP{} or \DrDLP{} can be interpreted as a probability distribution over $\calM_{k}$, in which case we have
    \begin{align*}
        \sum_{W \in \calM_k} |D(W)|\cdot\bfq[W] &=   \underset{W \sim \bfq}{\Mean}\big[\left|D(W)\right|\big], \\
       \forall \vote \subseteq \alt : \sum_{\substack{{W \in \calM_k:} \\{|D(W) \cap \vote| > |W \cap \vote| }}}  \bfq[W] &= \underset{W \sim \bfq}{\Prob}\big[|W \cap \vote| < |D(W) \cap \vote| \big].
    \end{align*}
    Plugging these into $\DLP$, it is clear that any optimal assignment $(u^*,\bfq^*)$ to the program must satisfy $$u^* = \max_{\vote \subseteq
     \alt} \underset{W \sim \bfq^*}{\Prob}\big[|W \cap \vote| < |D(W) \cap \vote| \big].$$
    With this in mind, the corollary statement follows from \Cref{lemma:core linear}, \Cref{def:droopcore}, and \Cref{thm:dual_eq}, which yield that the core (resp.\ Droop core) is non-empty for all profiles / vote distributions if and only if $\DLP < 0$ (resp.\ $\DrDLP \leq 0$) for all deviation functions $D$. 
\end{proof}

\subsection{Connection to \citet{Cheng19:Group}}\label{appsec:connection}

Working in the same setting as us, \citet[Sec. 2.3]{Cheng19:Group} show that, given any $m$ and $k$, for every distribution $\bfr \in \Delta(\calM_{\leq k})$ over deviations, there exists a distribution $\bfq \in \Delta(\calM_k)$ over $k$-committees (called a \emph{stable lottery}) such that for all votes $\vote \subseteq \alt$, we have 
 \begin{equation} \label{eq:stablecom}
     \underset{W \sim \bfq, W' \sim \bfr}{\Prob}\big[|W \cap \vote| < |W' \cap \vote| \big] < \, \frac{ \underset{W' \sim \bfr}{\Mean}\big[\left|W' \right|\big]}{k }.
 \end{equation}
This statement differs from the boxed statement in our \Cref{cor:prob} by the fact that we are given not a deviation function but a distribution over deviations. This allows \citet{Cheng19:Group} to use the independence of $W$ and $W'$ in their analysis, which we cannot employ (as $D(W)$ crucially depends on $W$). Nevertheless, we are hopeful that \Cref{cor:prob} may pave the way for a probabilistic analysis towards proving core non-emptiness (of deterministic committees), possibly similar to \citeauthor{Cheng19:Group}'s approach for stable lotteries. 

It is also worth pointing out that \citet{Cheng19:Group} arrive at \eqref{eq:stablecom} by observing that whether there exists a core-stable committee for a given approval profile $\profile = (A_i)_{i \in N}$ can be interpreted as a two-player zero-sum game: first, the minimizer chooses a committee $W \in \calM_k$; the maximizer then chooses a deviation $W' \in \calM_{\leq k}$ and gets utility
\begin{align*}
     |\{i \in N: |\vote_i \cap W'| > |\vote_i \cap W| \}|- n\frac{|W'|}{k}.
\end{align*}
Then, the profile admits no core-stable committee iff the value of the game is non-negative. If we instead have the maximizer and the minimizer choose an action simultaneously, the value of the game is nonnegative iff the profile admits a  stable lottery.

Following a similar reasoning (and using \Cref{cor:prob}), our version of the optimization problem given by \MMM{} (where we also search over all vote distributions $\bfx$ for a given $m,k$) can be interpreted as an (normal-form) \emph{adversarial team game}, where the maximizer is now a team of two players: one team member that picks a specific vote $\vote \subseteq \alt$, and another team member that picks a deviation function $D: \calM_k \rightarrow \calM_{\leq k}$, while the minimizer still picks a committee $W \in \calM_k$. If we set the utility of the maximizer team to be $\1[|D(W) \cap \vote | > |W \cap \vote |]- \frac{|D(W)|}{k}$ (where $\1$ is the indicator function), it can be shown that their \emph{Team-Maxmin payoff}~\citep{vonStengel97:Team} is nonnegative iff \MMM{} is nonnegative. 

If the two team members could correlate their strategies, their expected utility would be trivially nonnegative for sufficently large $m$. Whether \MMM{} is nonnegative (corresponding to when the two maximizer cannot correlate), is then directly related to the \emph{price of uncorrelation} of this adversarial team game~\citep{Basilico17:Team}. Equivalently, it is directly related to the \emph{value of recall}~\citep{Berker25:Value} if the maximizers are interpreted as a single player with imperfect recall. We believe these additional perspectives on the core non-emptiness problem can lead to valuable insights.

\subsection{Proof of \Cref{thm:m=k+1}}
\kplusone*
\begin{proof}
    Fix $m$, and say $k=m-1$ and $\alt=\{c_1,c_2,\ldots,c_m\}$. For convenience, we refer to each $k$-committee by the single candidate they do not contain; \emph{i.e.}, we write $\calM_k = \{W^{(j)} \}_{j \in [m]}$ where $W^{(j)} = \alt \setminus \{c_j\}$. We consider two cases:
    
    \noindent\textbf{\emph{Case 1:}} There is some $j \in [m]$ such that $D(W^{(j)}) \subseteq W^{(j)}$. In that case, set  $\bfq[W^{(j)}]=1$ and $\bfq[W]=0$ for all other $W$. Similarly, set $u=0$. This is a feasible solution to both \DLP{} and \DrDLP{}, as $\underset{\substack{{W \in \calM_k:} \\{|D(W) \cap \vote| > |W \cap \vote| }}}{\sum} \bfq[W]=0=u$ for all $\vote \subseteq \alt$. The objective of \DLP{} for this assignment is
    \begin{align*}
 u - \sum_{W \in \calM_k} \frac{|D(W)|}{k}\bfq[W]   = -\frac{|D(W^{(j)})|}{k} \leq -\frac{1}{k} < -\frac{1}{k(k+1)} .  \end{align*}
 For \DrDLP{}, respectively, the objective is 
    \begin{align*}
 u - \sum_{W \in \calM_k} \frac{|D(W)|}{k+1}\bfq[W]   = -\frac{|D(W^{(j)})|}{k+1} \leq -\frac{1}{k+1} <0. \end{align*}
Since both programs involve minimizing the objective, this assignment gives us the desired upper bounds from the theorem statement.

 \noindent\textbf{\emph{Case 2:}} There is no $j \in [m]$ such that $D(W^{(j)}) \subseteq W^{(j)}$. Equivalently, we have $c_j \in D(W^{(j)})$ for all $j \in [m]$. We first prove a useful lemma:
    \begin{lemma}\label{lemma:kplusone}
        Given any vote $\vote \subseteq \alt$ and candidate $c_j$ for some $j \in [m]$, we have $|D(W^{(j)}) \cap \vote| > |W^{(j)} \cap \vote|$ if and only if both (1) $c_j \in \vote$ and (2) $\vote \subseteq D(W^{(j)})$ are satisfied.
    \end{lemma}
    \begin{proof}[Proof:]
        \noindent $(\Rightarrow)$: Assume $|D(W^{(j)}) \cap \vote|> |W^{(j)} \cap \vote|$. Then we have $|\vote| \geq |D(W^{(j)}) \cap \vote| > | W^{(j)} \cap \vote | = |\vote \setminus \{c_j\}|$, so $c_j \in \vote$.  Moreover, this implies we have $|\vote| \geq |D(W^{(j)}) \cap \vote| > |\vote\setminus\{c_j\}| = |\vote|-1$, so $|\vote| =|D(W^{(j)}) \cap \vote|$ and hence $\vote \subseteq D(W^{(j)})$. 
        
        \noindent $(\Leftarrow)$: Assume (1) and (2) from the right hand side of the lemma statement. Then, we have $|D(W^{(j)}) \cap \vote|=|\vote|>|\vote|-1=|\vote \setminus \{c_j\}| = |\vote \cap W^{(j)}|$, where the first and second equalities follow from (2) and (1), respectively.
    \end{proof}
    
    Now, we are ready to give a feasible assignment to \DLP{} and \DrDLP{} that proves the theorem statement. We do a similar construction to the proof of \Cref{thm:singleton}:
    \begin{itemize}
        \item Step 1: Fix some arbitrary $W_1 \in \calM_k$ (Note that we use subscripts for the step $i=1,2,...$, and superscripts for candidates that are not included, \emph{i.e.}, $W_1=W^{(j)}=\alt \setminus \{c_j\}$ for some $j \in [m]$). 
        \item  Step $i=2,3,\ldots$: If $| \cup_{i' \in [i-1]} D(W_{i'})| = m$, then terminate. Otherwise, there is some candidate $c_j \notin \cup_{i' \in [i-1]} D(W_{i'})$. Set $W_i=W^{(j)}$. 
    \end{itemize}
    This process must terminate after fixing $\{W_i\}_{i \in [T]}$ for some $T \leq m$ steps, since at each step we increase $|\cup_{i' \in [i-1]} D(W_{i'})|$ by at least 1 by the premise in Case 2, and since $|\cup_{i' \in [i-1]} D(W_{i'})|$ is upper bounded by $m$. Set $\bfq[W]= \begin{cases} \frac{1}{T}&\text{if }W \in \{W_i\}_{i \in [T]} \\ 0 &\text{otherwise}\end{cases}$ and $u= \frac{1}{T}$. We will show that this assignment is feasible. Fix any $\vote \subseteq \alt$. We claim $|D(W_i) \cap \vote|> |W_i \cap \vote|$ for at most one $i \in [T]$, which shows the second set of constraints in \DLP{} is satisfied. Say $i^*$ is the smallest $i \in [T]$ such that $|D(W_i) \cap \vote|> |W_i \cap \vote|$ (If no such $i$ exists, we are done). Since $W_i \in \calM_k$ by construction, we must have $W_{i^*}=W^{(j^*)}$ for some $j^* \in [m]$. By \Cref{lemma:kplusone}, we have $c_{j^*} \in \vote$ and $\vote \subseteq D(W^{(j^*)})$. Take any $i \in \{i^*+1,i^*+2,\ldots,T\}$. Again, we must have $W_{i}=W^{(j)}$ for some $j \in [m]$. By construction, we have $c_{j} \notin \cup_{i' \in [i-1]}D(W_{i'}) \Rightarrow c_j \notin D(W_{i^*})$. Since $\vote \subseteq D(W^{(j^*)})=D(W_{i^*})$, this implies $c_j \notin \vote$, and hence we cannot have $|D(W^{(j)}) \cap \vote| > |W^{(j)} \cap \vote|$ by \Cref{lemma:kplusone}. We therefore have  $|D(W_i) \cap \vote|=|D(W^{(j)}) \cap \vote| \leq  |W^{(j)} \cap \vote| =|W_i \cap \vote|$, proving our claim. Hence, $u=\frac{1}{T}$ with the above $\bfq$ is a feasible assignment to \DLP{} (and \DrDLP). Further, we have $\sum_{W \in \calM_k} |D(W)| \cdot \bfq[W] = \frac{1}{T} \sum_{i \in [T]} |D(W_i)| \geq  \frac{1}{T}  | \cup_{i \in [T]} D(W_i)|=\frac{m}{T}$. Therefore for any $\ell \in \{k,k+1\}$, we have $ u - \sum_{W \in \calM_k} \frac{|D(W)|}{\ell}\bfq[W]  \leq \frac{1}{T} - \frac{m}{T\ell} =-\frac{m-\ell}{T\ell} \leq -\frac{m-\ell}{m\ell}= -\frac{k+1-\ell}{(k+1)\ell}$. Plugging in $\ell=k$ (resp.\ $\ell=k+1$) gives the desired upper bound for \DLP{} (resp.\ \DrDLP).
\end{proof}

\section{On \Cref{sec:lindahl}}

\subsection{Relationship between priceability axioms}\label{appsec:pbility}

In this subsection, we clarify the relationship between various notions of priceability introduced in prior literature, namely \emph{priceability} by \citet{Peters20:Proportionality}, as well as \emph{Lindahl priceability} and \emph{weak priceability} by \citet{Munagala22:Auditing}. Note that to stay consistent with the model of these papers, we represent elections with approval profiles $\profile=(\vote_i)_{i \in N}$ in this subsection, although these axioms can also be stated using our model of vote distributions (\emph{e.g.}, see \Cref{lemma:pbility-linear}).

First, we recall the definition of Lindahl priceability from \Cref{sec:lindahl}.

\lindahlpbility*

Weak priceability, also introduced by \citet{Munagala22:Auditing}, can be defined in a similar way, with only changing condition 2 of \Cref{def:lindahl}.

\begin{defn}\label{def:wpbility}
    A $k$-committee $W$ is \emph{weakly priceable} with respect to profile $\profile$ if $\exists$ a price system 
    $\{\bfp[i,c]\}_{i \in N, c \in \alt} \geq 0$ such that
    \begin{enumerate}
        \item $\forall c \in \alt$: $\sum_{i \in N} \bfp[i,c] \leq \frac{n}{k}$, and
        \item $\forall i \in N, d \in A_i \setminus W: \bfp[i,d]+\underset{c \in A_i \cap W}{\sum} \bfp[i,c] > 1$.
    \end{enumerate}

\end{defn}

It is clear that Lindahl priceability implies weak priceability. The reverse implication is false, as evidenced by the fact that Lindahl priceability implies core stability, whereas there are various voting rules (\emph{e.g.}, the Phr{\'a}gmen rule) that satisfy weak priceability but fail core stability~\citep{Munagala22:Auditing}. Next, we introduce an earlier definition by \citet{Peters20:Proportionality}.

\begin{defn}\label{def:pbility}
    A $k$-committee $W$ is \emph{priceable} with respect to profile $\profile$ if $\exists$ a pair $(r, \{f_i\}_{i \in N})$ consisting of a \emph{price} $r > 0$ and, for each voter $i \in N$, a  \emph{payment function} $f_i\colon \alt \to [0, 1]$, such that
    \begin{enumerate}
        \item If $f_i(c) > 0$, then $c \in A_i$ (a voter can only pay for candidates she approves of),
        \item For each $i\in N:$ $\sum_{c \in  \alt} f_i(c) \leq 1$ (a voter can spend at most one dollar) \text{.}
        \item For each elected candidate $c \in W$, the sum of the payments to this candidate equals the price: $\sum_{i \in N}f_i(c) = r$.
        \item No candidate outside of the committee $c \in \alt \setminus W$ gets any payment: $\sum_{i \in N}f_i(c) = 0$.
        \item For each candidate outside of the committee $c \in \alt \setminus W$, the remaining unspent budget of the voters that approve that candidate is at most $r$:
        \begin{align*}
            \sum_{i \in N\colon c\in \vote _i} \left(1 - \sum_{c' \in W}f_i(c')\right) \leq r \text{.}
        \end{align*}
    \end{enumerate}
\end{defn}

We first show that Lindahl priceability and priceability are incomparable. Similar to weak priceability, the fact that priceability does not imply Lindahl priceability follows from the existence of voting rules (\emph{e.g.}, the Phr{\'a}gmen rule) that satisfy priceability but fail core stability~\citep{Peters20:Proportionality}. We now formalize the other direction.

\begin{prop}\label{prop:lindnotpble}
    There exists an approval profile such that there is a committee that is Lindahl priceable (and therefore weakly priceable) but not priceable.
\end{prop}
\begin{proof}
    Say $n=5,k=2$. The candidates are $\alt=\{a,b,d\}$, and  the approval sets are $A_1=\{a\}$ and $A_i=\{b,d\}$ for all $i \neq 1$. The committee is $W=\{a,b\}$.

    $W$ is Lindahl priceable (\Cref{def:lindahl}): setting $\bfp[i,b] = 0.6$ and  $\bfp[i,d] = 0.5$ for all $i  \neq 1$, and $\bfp[i,c] = 0$ for all other voter-candidate pairs satisfies both constraints. In particular, the left hand side of condition 1 of \Cref{def:lindahl} for $a$, $b$, and $d$ are 0, 2.4, and 2, respectively, all of which is less than $n/k=2.5$. Condition 2 is only relevant for $i  \neq 1$ and $T=\{b,d\}$, for which we get $0.5+0.6>1$ as desired.
    
    However, $W$ is not priceable (\Cref{def:pbility}). Assume for contradiction that there exists a pair $(r, \{f_i\}_{i \in N})$ that satisfies all the conditions. Then conditions 1, 2, and 3 gives us $1 \geq \sum_{c\in \alt} f_1(c) = f_1(a) = \sum_{i\in N} f_i(a)= r$. Similarly, we have $r= \sum_{i  \neq 1} f_i(b)$. However, invoking condition 5 for $d$ then gives us $r \geq \sum_{i\neq 1} (1- f_i(b)) = 4 - r$, so $r \geq 2$, a contradiction.
\end{proof}

We now compare priceability by \citet{Peters20:Proportionality} with weak priceability by \citet{Munagala22:Auditing}. \Cref{prop:lindnotpble} already establishes that weak priceability does not imply priceability. We now prove that the other direction holds, \emph{i.e.} priceability does imply weak priceability.

\begin{prop}
    Any priceable committee is also weakly priceable.
\end{prop}
\begin{proof}
    Say $W$ is priceable, as certified by a pair $(r, \{f_i\}_{i \in N})$ satisfying the conditions of \Cref{def:pbility}. Note that by conditions 1-4, we have
    \begin{align*}
        n \geq \sum_{i \in N} \sum_{c \in \alt} f_i(c) = \sum_{c \in W} \sum_{i \in N} f_i(c) = kr \Rightarrow \frac{n}{k} \geq r
    \end{align*}

    If $\frac{n}{k}<r$, for each $c \in \alt \setminus W$ and $i \in N$ such that $c \in A_i$, pick a $\epsilon_{i,c}>0$ small enough such that $\sum_{i \in N: c \in A_i} \epsilon_{i,c} \leq \frac{n}{k}-r$. Otherwise (if $\frac{n}{k}=r$), set $\epsilon_{i,c}=\frac{1}{k}$ for all such $i,c$.

    To prove that $W$ is weakly priceable, we set $\{\bfp[i,c]\}_{i \in N, c \in \alt}$ from \Cref{def:wpbility} as follows:
    \begin{itemize}
        \item For all $c \in W, i \in N$, set $\bfp[i,c]=f_i(c)$.
        \item For all $c \in \alt \setminus W, i \in N$ s.t.\ $c \in A_i$, set $\bfp[i,c]=\left(1-\sum_{c' \in W} f_i(c')\right)+\epsilon_{i,c}$
        \item For all other pairs of $i \in N$ and $c \in \alt$, set $\bfp[i,c]=0$.
    \end{itemize}
    We claim $\bfp$ satisfies both conditions of \Cref{def:wpbility}. First, we show it satisfies condition 1:
    \begin{itemize}
        \item $\forall c \in W: \sum_{i \in N} \bfp[i,c]= \sum_{i \in N}f_i(c) = r \leq \frac{n}{k}$, where the second equality follows from condition 3 of \Cref{def:pbility}.
        \item $\forall c \in \alt \setminus W:$ If $\frac{n}{k}>r$, we have $\sum_{i \in N} \bfp[i,c]= \sum_{i \in N: c \in A_i} \left(1-\sum_{c' \in W} f_i(c')\right)+\epsilon_{i,c}\leq r + \sum_{i \in  N: c \in A_i} \epsilon_{i,c} \leq r+\frac{n}{k}-r= \frac{n}{k}$, where the first inequality follows from condition 5 of \Cref{def:pbility}. If $\frac{n}{k}=r$, condition 2 of \Cref{def:pbility} must be an equality for each $i \in N$, \emph{i.e.}, $1=\sum_{c \in \alt} f_i(c) = \sum_{c \in W} f_i(c)$, where the second inequality follow from condition 4 of \Cref{def:pbility}. Then we have $\sum_{i \in N} \bfp[i,c]= \sum_{i \in N: c \in A_i} \left(1-\sum_{c' \in W} f_i(c')\right)+\epsilon_{i,c}= \sum_{i \in N: c \in A_i} 0+\epsilon_{i,c} = \frac{|\{i \in N: c \in A_i\}|}{k} \leq \frac{n}{k}$.
    \end{itemize}

    Finally, we show that our choice of $\bfp$ satisfies condition 2 of \Cref{def:wpbility}:
    \begin{itemize}
        \item $\forall i \in N, d \in A_i \setminus W: \bfp[i,d]+\underset{c \in A_i \cap W}{\sum} \bfp[i,c]= \left(1-\sum_{c' \in W} f_i(c')\right)+\epsilon_{i,d} + \sum_{c \in A_i \cap W} f_i(c) = 1+\epsilon_{i,d} >1$, where the last equality follows from condition 1 of \Cref{def:pbility}.
    \end{itemize}
\end{proof}

This concludes our discussion on the relationship between  different priceability axioms.\footnote{On another note, \citet{Peters21:Market} has studied the axiom \emph{stable priceability}, which implies priceability and core stability. Stable priceable committees, however, are known to sometimes not exist, whereas this is open for Lindahl priceability.}

\subsection{Proof of \Cref{lemma:pbility-linear}}
\linprice*

\begin{proof}
Fix $m,k,n$, profile $\profile$, and a $k$-committee $W$. For any $\vote \subseteq \alt$, define $\calT_{W}(\vote)= \{T \subseteq \alt: |\vote \cap T| > |\vote \cap W|\}$. We will prove the statement for Lindahl priceability; the proof for weak priceability follows identically, except we would have instead defined $\calT_{W}(\vote)= \{T \subseteq \alt: T=\{d\}\cup(\vote \cap W) \text{ for some }   d \in \vote \setminus W\}$. We prove the lemma in two steps.

\noindent\textbf{\emph{Step 1:}} First, we prove that $W$ is Lindahl priceable \wrt $\profile$ if and only if $\exists$ a price system 
    $\{\bfp[i,c]\}_{i \in N, c \in \alt} \geq 0$ such that
    \begin{enumerate}
        \item[1.] $\forall c \in \alt$: $\sum_{i \in N} \bfp[i,c] \leq \frac{n}{k}$,
        \item[2.] $\forall i \in N, T\in \calT_{W}(A_i): \underset{c \in T}{\sum} \bfp[i,c] > 1$, and
        \item[3.] $\forall i,j \in N$, $A_i=A_j \Rightarrow \bfp[i,c]=\bfp[j,c]$ for all $c \in \alt$.
\end{enumerate}
The ``if'' direction is trivial, since the first two conditions are sufficient for Lindahl priceability by \Cref{def:lindahl}. For the ``only if'' direction, assume $W$ is Lindahl priceable, implying there exists a price system $\{\bfp'[i,c]\}_{i \in N, c \in \alt} \geq 0$ that satisfies the first two conditions. We will construct an alternative price system $\bfp'$ that satisfies all three conditions. For any $\vote \subseteq \alt$, define $N(\vote)=\{i \in N: A_i =A\}$. Then, for any $i \in N$ and $c \in \alt$, we define $\bfp'[i,c]=\frac{1}{|N(A_i)|} \sum_{j \in N(A_i)}\bfp[j,c]$, \emph{i.e.}, each voter's price for a candidate is now the average over the (former) prices for that candidate of the voters with the same approval set. This new price system $\{\bfp'[i,c]\}_{i \in N, c \in \alt}$ clearly satisfies condition 3. It also satisfies condition 1, because the total payment to any candidate has not changed. For the second condition, we have $\forall i \in N, T\in \calT_{W}(A_i)$, 
\begin{align*}
    \underset{c \in T}{\sum} \bfp'[i,c] =  \underset{c \in T}{\sum} \frac{1}{|N(A_i)|}\underset{j \in N(A_i)}{\sum} \bfp[j,c] = \frac{1}{|N(A_i)|}\underset{j \in N(A_i)}{\sum} \underset{c \in T}{\sum}  \bfp[j,c] > \frac{1}{|N(A_i)|}\underset{j \in N(A_i)}{\sum} 1 = 1, 
\end{align*}
since $\calT_{W}(A_j)=\calT_{W}(A_i)$ for all $j \in N(A_i)$, so $\bfp$ must satisfy condition 2 for each $j \in N(A_i)$ and $T \in \calT_{W}(A_i)$ by assumption. 

\noindent \textbf{\emph{Step 2:}} By Step 1, we know that \wlogg{} we can restrict price systems in the definition of Lindahl priceability (\Cref{def:lindahl}) to those that assign the same prices to any $i,j \in N$ such that $A_i =A_j$. Equivalently, this implies that $W$ is Lindahl priceable \wrt $\profile$ if and only if $\exists$ a price system 
    $\{\bfp[A,c]\}_{A \subseteq \alt, c \in \alt} \geq 0$ such that
        \begin{enumerate}
        \item[1.] $\forall c \in \alt$: $\sum_{\vote \subseteq \alt} |N(A)| \cdot \bfp[\vote,c] \leq \frac{n}{k}$, and
        \item[2.] $\forall \vote \subseteq \alt, T\in \calT_{W}(A): \underset{c \in T}{\sum} \bfp[A,c] > 1$.
\end{enumerate}
Dividing both sides of the first condition by $n$ and noting $\bfx[\vote] = \frac{|N(\vote)|}{n}$ (where $\bfx$ is the vote distribution of $\profile$) gives us the lemma statement.
\end{proof}

\subsection{Dual formulation of weak/Lindahl priceability}\label{appsec:dualpbility}

Fix $m,k,n$, profile $\profile$, and a $k$-committee $W$. For any $\vote \subseteq \alt$, define $\calT_{W}(\vote)= \{T \subseteq \alt: |\vote \cap T| > |\vote \cap W|\}$ again. Every statement we make in this subsection about Lindahl priceability is true for weak priceability---following the definition by \citet{Munagala22:Auditing}---except we would have instead defined $\calT_{W}(\vote)= \{T \subseteq \alt: T=\{d\}\cup(\vote \cap W) \text{ for some }   d \in \vote \setminus W\}$. 

Due to \Cref{lemma:pbility-linear}, we have that $W$ is Lindahl priceable if and only if there exists a price system $\bfp = \{ \bfp[\vote,c] \}_{\vote \subseteq \alt, c \in \alt} \geq 0$ such that
    \begin{enumerate}
        \item[1.] $\forall c \in \alt$: $\sum_{\vote \subseteq \alt} \bfx[\vote] \cdot \bfp[\vote,c] \leq \frac{1}{k}$, and
        \item[2.] $\forall \vote \subseteq  \alt, T \in \calT_{W}(A): \underset{c \in T}{\sum} \bfp[\vote,c] > 1$
    \end{enumerate}
    Equivalently, $W$ is Lindahl priceable if and only if the value of the following linear program (recall that $\bfx$ is fixed) over variables $\{\bfp[A,c]\}_{\vote \subseteq \alt, c \in \alt}$ is strictly greater than 1:
\begin{align}
 &\max_{\bfp \in \R^{2^\alt \times \alt}, \, \varepsilon \in \R} \quad  \varepsilon \quad  \tag{\LinLP}
 \\
 &\textnormal{s.t.} \, \, \forall \vote \subseteq \alt, c \in \alt : \bfp[\vote, c] \geq 0 \nonumber
\\
&\quad \, \forall c \in \alt : \sum_{\vote \subseteq \alt} \bfx[\vote] \cdot \bfp[\vote,c] \leq \frac{1}{k} \nonumber
\\
&\quad \, \forall \vote \subseteq  \alt, T \in \calT_{W}(A): \underset{c \in T}{\sum} \bfp[\vote,c] \geq \varepsilon \nonumber
\end{align}

Equivalently, $W$ is \textbf{not} Lindahl priceable if and only if $\LinLP \leq 1$. It is straightforward to show that additionally imposing $\bfp[\vote, c] =0$ for all $\vote \subseteq \alt$ and $c \in \alt \setminus \vote$ does not hurt the objective of \LinLP, effectively bringing down variables to  $\{\bfp[A,c]\}_{\vote \subseteq \alt, c \in \vote}$. By using strong duality, this implies that $W$ is not Lindahl priceable  if and only if the value of the following program is at most 1:
\begin{align}
 &\min_{\bft \in \R^{\alt}, \, \bfg \in \R^{2^\alt \times 2^\alt }} \quad  \frac{1}{k} \sum_{c \in \alt} \bft[c]  \quad  \tag{\LinDLP}
 \\
 &\textnormal{s.t.} \, \, \forall  c \in \alt : \bft[c] \geq 0 \nonumber
\\
&\quad \,  \forall  \vote \subseteq \alt, T \in \calT_{W}(A) : \bfg[\vote,T] \geq 0 \nonumber
\\
&\quad \,  \sum_{\vote \subseteq \alt} \sum_{T \in \calT_{W}(\vote)} \bfg[\vote,T]=1 \nonumber
\\
&\quad \, \forall \vote \subseteq  \alt,  c \in A: \bft[c] \cdot x[\vote]\geq \underset{T \in \calT_{W}(A) \, : \, c \in T}{\sum} \bfg[\vote,T]  \nonumber
\end{align}

Then, in order to find a vote distribution $\bfx$ and a $k$-committe $W^*$ such that $W^*$ is in the (Droop) core but not Lindahl priceable, we can simply solve the following integer quadratic program (\wlogg{} fix $W^*=\{c_1,c_2,\ldots,c_k\}$):
\begin{align}
 &\min_{\bfx \in \R^{2^\alt}, \, \mu \in \R, \, \bft \in \R^{\alt}, \, \bfg \in \R^{2^\alt \times 2^\alt }} \quad  \mu \quad  \tag{\LinQIP}
 \\
 &\textnormal{s.t.} \, \, \sum_{\vote \subseteq \alt} \bfx[\vote] = 1 \quad \text{and} \quad \forall \vote \subseteq \alt : \bfx[\vote] \geq 0 \nonumber
\\
&\quad \, \frac{1}{k} \sum_{c \in \alt} \bft[c] 
 \leq 1 \nonumber
\\
&\quad \, \forall  c \in \alt : \bft[c] \geq 0 \nonumber
\\
&\quad \,  \forall  \vote \subseteq \alt, T \in \calT_{W^*}(A) : \bfg[\vote,T] \geq 0 \nonumber
\\
&\quad \,  \sum_{\vote \subseteq \alt} \sum_{T \in \calT_{W^*}(\vote)} \bfg[\vote,T]=1 \nonumber
\\
&\quad \, \forall \vote \subseteq  \alt,  c \in A: \bft[c] \cdot x[\vote]\geq \underset{T \in \calT_{W^*}(A) \, : \, c \in T}{\sum} \bfg[\vote,T]  \nonumber
\\
&\quad \, \forall W' \in \calM_{\leq k} :  \mu \geq \impr_{W^*,W'}^T \bfx - \frac{|W'|}{k} \nonumber
\end{align}

\LinQIP{} incorporates elements from both \MILP{} and \LinDLP. The constraints from \LinDLP{} (along with the additional constraint forcing its objective to be at most 1) ensures that \LinQIP{} is restricted to profiles for which $W^*$ is not Lindahl priceable. If $\LinQIP<0$, then the program has found a profile such that $W^*$ is additionally in the core: for all deviations $W'$, we have $\impr_{W^*,W'}^T \bfx - \frac{|W'|}{k} < 0$. If $\LinQIP \geq 0$, on the other hand, this implies no such profile exists, \emph{i.e.}, every core-stable committee is pricable for these values of $m$ and $k$. Replacing the $\frac{|W'|}{k}$ in the last constraint of \LinQIP{} with $\frac{|W'|}{k+1}$ and checking if $\LinQIP \leq 0$ (rather than $\LinQIP<0$), on other hand, is equivalent to checking if there is a profile such that $W^*$ is Droop core-stable but not Lindahl priceable. Using the final assignments for $\bft$ and $\bfg$ (which are the variables from \LinDLP), we can construct human-readable proofs for \Cref{thm:no-imply} below, as these variables give us exactly the coefficents with which we need to multiply the associated constraints in \LinLP{} in order to show $W^*$ is not Lindahl priceable.

\subsection{Proof of \Cref{thm:no-imply}}
\noimply*
\begin{proof}
    We first prove the claim for weak priceability. Fix $m=5,k=3$. Consider the the vote distribution $\bfx$ defined as
    \begin{align*}
        \bfx[\vote]= \begin{cases}
            1/4 &\text{if }\vote=\{c_2,c_4\}\\
            1/4 &\text{if }\vote=\{c_2,c_5\}\\
            1/4 &\text{if }\vote=\{c_4,c_5\}\\
            1/4 &\text{if }\vote=\{c_1,c_2,c_5\}\\
            0 & \text{otherwise}
        \end{cases}
    \end{align*}
    The committee $W=\{c_1,c_2,c_3\}$ is Droop core-stable. For the sake of contradiction, assume $W$ is also weakly priceable, certified by a price system $\bfp =$ $\{\bfp[\vote,c]\}_{\vote \subseteq \alt, c \in \alt}$. Condition 1 of \Cref{lemma:pbility-linear} then implies
\begin{align*}
     \frac{\bfp[\{c_2,c_4\}, c_2]}{4} + \frac{\bfp[\{c_2,c_5\}, c_2]}{4} &\leq  \sum_{\vote \subseteq \alt} \bfx[\vote]  \bfp[\vote,c_2] \leq \frac{1}{3}\\
     \frac{\bfp[\{c_2,c_4\}, c_4]}{4} + \frac{\bfp[\{c_4,c_5\}, c_4]}{4} &\leq  \sum_{\vote \subseteq \alt} \bfx[\vote]  \bfp[\vote,c_4] \leq \frac{1}{3}\\
     \frac{\bfp[\{c_2,c_5\}, c_5]}{4} + \frac{\bfp[\{c_4,c_5\}, c_5]}{4} &\leq  \sum_{\vote \subseteq \alt} \bfx[\vote]  \bfp[\vote,c_5] \leq \frac{1}{3}
\end{align*}
Similarly, condition 2 of \Cref{lemma:pbility-linear} yields
\begin{align}
    &\bfp[\{c_4,c_5\}, c_4] >1 \quad \text{and} \quad \bfp[\{c_4,c_5\}, c_5]  >1 \label{eq:no-a1}\\
    &\bfp[\{c_2,c_4\}, c_2] + \bfp[\{c_2, c_4\}, c_4] >1 \label{eq:no-a2}\\
    &\bfp[\{c_2,c_5\}, c_2] + \bfp[\{c_2, c_5\}, c_5]  >1 \label{eq:no-a3}
\end{align}
If we multiply \eqref{eq:no-a1}-\eqref{eq:no-a3} by $-1/4$ and add them up together with the previous three inequalities from Condition 1, we get the contradiction $0<0$. 

We now prove the theorem statement for Lindahl priceability. Fix $m=4,k=2$. Consider the the vote distribution $\bfx$ defined as
    \begin{align*}
        \bfx[\vote]= \begin{cases}
            1/3 &\text{if }
            \vote=\{c_3,c_4\}\\
            1/3 &\text{if }\vote=\{c_1,c_3,c_4\}\\
            1/3 &\text{if }\vote=\{c_1,c_2,c_3,c_4\}\\
            0 & \text{otherwise}
        \end{cases}
    \end{align*}
    The committee $W=\{c_1,c_2\}$ is Droop core-stable. For the sake of contradiction assume $W$ is also Lindahl priceable, certified by a price system $\bfp =$ $\{\bfp[\vote,c]\}_{\vote \subseteq \alt, c \in \alt}$. Condition 1 of \Cref{lemma:pbility-linear} then implies
\begin{align*}
   \frac{\bfp[\{c_1,c_3,c_4\}, c_1]}{3}+ \frac{\bfp[\{c_1,c_2,c_3,c_4\}, c_1]}{3}  &\leq  \sum_{\vote \subseteq \alt} \bfx[\vote]  \bfp[\vote,c_1] \leq \frac{1}{2}\\
   \frac{\bfp[\{c_3,c_4\}, c_3]}{3}+  \frac{\bfp[\{c_1,c_3,c_4\}, c_3]}{3}+ \frac{\bfp[\{c_1,c_2,c_3,c_4\}, c_3]}{3}   &\leq  \sum_{\vote \subseteq \alt} \bfx[\vote]  \bfp[\vote,c_3] \leq \frac{1}{2}\\
   \frac{\bfp[\{c_3,c_4\}, c_4]}{3}+  \frac{\bfp[\{c_1,c_3,c_4\}, c_4]}{3}+ \frac{\bfp[\{c_1,c_2,c_3,c_4\}, c_4]}{3}   &\leq  \sum_{\vote \subseteq \alt} \bfx[\vote]  \bfp[\vote,c_4] \leq \frac{1}{2}.
\end{align*}
Similarly, condition 2 of \Cref{lemma:pbility-linear} yields
\begin{align}
    &\bfp[\{c_3,c_4\}, c_3] >1 \quad \text{and} \quad \bfp[\{c_3,c_4\}, c_4]  >1 \label{eq:no-b0}\\
    &\bfp[\{c_1,c_3,c_4\}, c_1] + \bfp[\{c_1,c_3,c_4\}, c_3] >1 \label{eq:no-b1}\\
    &\bfp[\{c_1,c_3,c_4\}, c_1] + \bfp[\{c_1,c_3,c_4\}, c_4]  >1 \label{eq:no-b2}\\
    &\bfp[\{c_1,c_3,c_4\}, c_3] + \bfp[\{c_1,c_3,c_4\}, c_4] >1 \label{eq:no-b3}\\
    &\bfp[\{c_1,c_2,c_3,c_4\}, c_1] +\bfp[\{c_1,c_2,c_3,c_4\}, c_3] +\bfp[\{c_1,c_2,c_3,c_4\}, c_4] >1 \label{eq:no-b4}.
\end{align}
If we multiply \eqref{eq:no-b1}-\eqref{eq:no-b3} by $-1/6$, multiply \eqref{eq:no-b0} and \eqref{eq:no-b4} by $-1/3 $, and add them up together with the previous three inequalities from Condition 1, we get the contradiction $0<0$.

\end{proof}

\subsection{Proof of \Cref{cor:linhdahl-no-imply}}\corepnotlp*

\begin{proof}
    We revisit a counterexample $\bfx$ from the proof of \Cref{thm:no-imply}. For $m=4$ and $k=2$, we had shown that $W=\{c_1,c_2\}$ is core-stable but not Lindahl priceable for the vote distribution $\bfx$ defined as
\begin{align*}
        \bfx[\vote]= \begin{cases}
            1/3 &\text{if }
            \vote=\{c_3,c_4\}\\
            1/3 &\text{if }\vote=\{c_1,c_3,c_4\}\\
            1/3 &\text{if }\vote=\{c_1,c_2,c_3,c_4\}\\
            0 & \text{otherwise}
        \end{cases}
    \end{align*}
    
    It remains to show that $W$ is weakly priceable. We can show this using the following price system $\{\bfp[\vote,c]\}_{\vote \subseteq \alt, c \in \alt}$, which satisfies the conditions of \Cref{lemma:pbility-linear} (for weak priceability).
    \begin{enumerate}
        \item[1.] $\bfp[\{c_3,c_4\},c_3] = \bfp[\{c_3,c_4\},c_4]= \bfp[\{c_1,c_3,c_4\},c_1]=\bfp[\{c_1,c_2,c_3,c_4\},c_2]=\frac{3}{2}$,
        \item[2.] $\bfp[\vote,c] =0$ for all \emph{other} vote-candidate pairs such that $\bfx[\vote]>0$,
        \item[3.] $\bfp[\vote,c] = 1.1$ for all vote-candidate pairs such that $\bfx[\vote]=0$.
    \end{enumerate}
\end{proof}

\end{document}